\documentclass{article}
\usepackage{ulem}
\usepackage[flushleft]{threeparttable}
\usepackage{microtype}
\usepackage{graphicx}
\usepackage{booktabs} 

\usepackage{hyperref}

\usepackage[accepted]{icml2025}
\usepackage{listings}
\usepackage{preamble}
\usepackage{notation}
\usepackage[utf8]{inputenc} 
\usepackage[T1]{fontenc}    
\usepackage{hyperref}       
\usepackage{url}            
\usepackage{booktabs}       
\usepackage{amsfonts}       
\usepackage{nicefrac}       
\usepackage{microtype}      
\usepackage{caption}

\icmltitlerunning{Prediction-aware Learning in Multi-agent Systems}

\begin{document}
\twocolumn[\icmltitle{Prediction-aware Learning in Multi-agent Systems}


\icmlsetsymbol{equal}{*}

\begin{icmlauthorlist}
\icmlauthor{Aymeric Capitaine}{polytechnique}
\icmlauthor{Etienne Boursier}{inriaorsay}
\icmlauthor{Eric Moulines}{polytechnique}
\icmlauthor{Michael I. Jordan}{inriaparis}
\icmlauthor{Alain Durmus}{polytechnique}
\end{icmlauthorlist}

\icmlaffiliation{polytechnique}{Centre de Mathématiques Appliquées – CNRS – École polytechnique – Palaiseau, 91120, France}
\icmlaffiliation{inriaorsay}{Inria Saclay, Université Paris Saclay, LMO - Orsay, 91400, France}
\icmlaffiliation{inriaparis}{Inria Paris, Ecole Normale Supérieure, PSL Research University - Paris, 75, France}

\icmlcorrespondingauthor{Aymeric Capitaine}{firstname.lastname@polytechnique.edu}

\icmlkeywords{Machine Learning, ICML}

\vskip 0.3in
]



\pagestyle{fancy}
\fancyhead{}
\fancyhead[C]{\small\bfseries Prediction-aware Learning in Multi-agent Systems}

\printAffiliationsAndNotice{} 

\begin{abstract}
The framework of uncoupled online learning in multiplayer games has made significant progress in recent years. In particular, the development of  \textit{time-varying games} has considerably expanded its modeling capabilities. However, current regret bounds quickly become vacuous when the game undergoes significant variations over time, even when these variations are easy to predict. Intuitively, the ability of players to forecast future payoffs should lead to tighter guarantees, yet existing approaches fail to incorporate this aspect. This work aims to fill this gap by introducing a novel \textit{prediction-aware} framework for time-varying games, where agents can forecast future payoffs and adapt their strategies accordingly. In this framework, payoffs depend on an underlying state of nature that agents predict in an online manner. To leverage these predictions, we propose the \texttt{POMWU} algorithm, a contextual extension of the optimistic Multiplicative Weight Update algorithm, for which we establish theoretical guarantees on social welfare and convergence to equilibrium. Our results demonstrate that, under bounded prediction errors, the proposed framework achieves performance comparable to the static setting. Finally, we empirically demonstrate the effectiveness of \texttt{POMWU} in a traffic routing experiment. 
\end{abstract}

\section{Introduction.}The framework of uncoupled online learning in multiplayer games has sparked a lot of interest for its ability to realistically model the interactions of rational players engaged in a dynamic game. Since the seminal works of \citet{foster1997, freund1999, hart2000}, progress has been made towards obtaining fast convergence rates for different equilibrium concepts, including coarse correlated equilibrium \citep{syrgkanis2015, foster2016, daskalakis2021, piliouras2022, farina2022clairvoyant} correlated equilibrium \citep{chen2020hedging, anagnostides2022near, anagnostides2022uncoupled, peng2023fastswapregretminimization} and Nash equilibrium \citep{anagnostides2022optimisticmirrordescentconverges}. However, most of these works assume that the game remains constant over time. 

Only recent studies have begun to consider time-varying games, in both two-player zero-sum games \citep{zhang2022no} and  multiplayer general-sum games \citep{duvocelle2018learning, anagnostides2024convergence}. Following methods initially developed by the online optimization community \citep{chiang2012online, rakhlin2013online}, these studies bound the dynamic regret incurred by players with measures of the \textit{time-variation} of the underlying game. While this approach looks satisfactory at first glance, it is not hard to come up with simple examples for which the variation is important--making the above mentioned bounds vacuous--yet very simple to predict. In \Cref{example:variation}, we exhibit a simple instance of time-varying game where the regret bounds derived in \citet{zhang2022no} grow linearly with the horizon $T>0$. However, the dynamic underlying the payoff matrices is entirely deterministic, and knowing it would result in a constant regret.
This highlights that the current time-varying framework fails to account for any \textit{predictive capacity} of the agents.  This is all the more surprising as predictive models become ubiquitous in numerous economic sectors \citep{jordan15review,papadimitriou21, jordan2024}, making it likely for strategic agents to possess a forecasting ability regarding their future payoffs. This work aims to fill the gap, by asking the following question:
\begin{quote}\label{question1}
    How does the quality of predictions made by rational agents in time-varying games regarding their future payoffs affects social welfare, as well as the convergence to equilibrium ? 
\end{quote}

\paragraph{Contributions.}
We address this question with the following contributions. \begin{itemize}
    \item First, we introduce the new \textit{prediction-aware} learning framework, where players forecast  future payoffs in an online fashion and design their strategies accordingly. In a nutshell, we build on the contextual setting proposed by \citet{sessa2021} by introducing an underlying state of nature, either adversarially or stochastically drawn, which determines the payoff of all agents. They play a time-varying game which can be decomposed into three stages. First, each player forecasts the current state of nature based on their local predictor before picking an action in the game. Then, they observe their payoff and the actual state of nature. Finally, they update their policy and predictor based on these new observations. Augmenting uncoupled learning in games with contexts and predictions requires to introduce new regret and equilibrium concepts. In particular, we extend  correlated equilibrium \citep{aumann1987correlated} to our framework. 
    \item Second, we propose an algorithm called \texttt{POMWU}---which a contextual extension of the optimistic Multiplicative Weight Update algorithm \citep{daskalakis2021}---allowing players to leverage their prediction about the state of nature. In particular, we show that if all players use \texttt{POMWU}, we match the results of \citet{syrgkanis2015} established for static games, regarding social welfare (\Cref{corollary:convergencewelfare}), equilibrium convergence (\Cref{corollary:convergenceeq}) and robustness in the adversarial setting (\Cref{proposition:robustness}) up to a factor that depends polynomially on the number of prediction errors by players. Thus, when predictions errors are bounded by a constant \citep[which is the case under realizability, see][]{daniely2014}, our bounds match the  guarantees on social welfare and equilibrium convergence for static games. 
    Our analysis builds upon a new notion of \textit{contextual} Regret bounded by Variation Utility (RVU) which bounds contextual regret by the sum of the length of the context-specific sequences of feedbacks and strategies. Indeed, a naive application of the standard RVU framework results in looser bounds.
\end{itemize}
\paragraph{Additional related works.}The problem tackled in this work relates with several lines of research in game theory and online optimization. On the one hand, the contextual optimization literature \citep{donti2019, elmachtoub2020, bennouna2024} has considered the problem of minimizing an objective function defined by an unobserved random context, which the optimizer can predict via a regression function. This idea has also been studied in the contextual bandit framework \citep{lattimore2020bandit} with noisy contexts \citep{kirschner2019stochasticbanditscontextdistributions, yang2021bandit, nelson2022linearizing, guo2024online}. However, none of these works consider the multi-agent setting, where the optimizer interacts with other agents during the learning process. On the other hand, recent studies in game theory have incorporated the idea of an underlying state of nature jointly determining the payoffs of players. While \citet{sessa2021, maddux2024} studies the contextual version of uncoupled learning in multiplayer games, \citet{lauffer2023no, harris2024regret} focuses on Stackelberg games with side information. However, these works assume that the context is revealed to players at the beginning of each period, unlike ours where players have to predict the context before moving. In the end, the \textit{social learning} framework might be the one that relates the most to ours. Pioneered by the work of \citet{banerjee1992simple, bikhchandani1992theory, smith2000pathological}, it features agents receiving private signals about a true, unobserved state of nature. These agents are able to learn from both their signal and the actions played by other players, which reflect their signals \citet{chamley2004rational}. Most of the social learning literature has been devoted to analyzing the resulting collective behaviors, such as cascading and herding phenomena \citep{mossel2020social}. While recent studies have broadened the analytical toolbox of social learning by considering for instance time-varying states of nature \citep{ frongillo2011,boursier22a, levy2024stationary}, it mostly relies on very strong assumptions \citep[e.g. a binary state and binary actions,][]{mossel2020social} and a Bayesian modeling where all agents share a common prior about the state of nature's distribution. In contrast, we believe that the uncoupled learning framework \citep{hart2000simple, hart2003uncoupled, daskalakis2011near} upon which our work relies is a more general setting for studying this question, and allows to study more natural equilibrium concepts such as correlated equilibria \citep{aumann1987correlated} with stronger guarantees.
\paragraph{Organization.}This work is organized as follows. In \Cref{section:model}, we present our model, notion of regret and main assumptions. In \Cref{section:predictionaware}, we introduce the \texttt{POMWU} algorithm and establish the convergence of social welfare and individual utilities. In \Cref{section:experiments}, we empirically demonstrate the performance of \texttt{POMWU} on the Sioux Falls routing problem \citep{leblanc1975efficient}. 

\begin{example}\label{example:variation}
    Consider the two-players setting in \cite{zhang2022no} where $\cX\in\R^n$ and $\cY\in\R^m$ are respectively the strategy spaces of player $x$ and $y$,   $A_t\in[-1,1]^{n\times m}$ is their time-varying payoff matrix and $\cE _t \subset \cX\times\cY$ is the set of Nash equilibria at time $t\in\timesteps$. The two measures of variations  considered in (\citealt{zhang2022no}, and up to  minor modifications \citealt{anagnostides2024convergence}) are
$$
P_T = \min_{\cE_1 \times\ldots\times\cE_T}\sum_{t\in\timesteps}\parenthese{\normone{x^\star _t - x^\star _{t-1}}+\normone{y^\star _t -y^\star _{t-1}}}\eqsp,$$
and
$$
V_T = \sum_{t\in\timesteps}\norminf{A_t - A_{t-1}}^2 \eqsp,
$$which are respectively the variation of Nash equilibria and the variation of payoff matrices.  \citet[][Theorem 6]{zhang2022no} show that the dynamic regret can be bounded by
\begin{equation}
    \label{example1:boundzhang}
    \widetilde{\bigO} (\min(\sqrt{(1+P_T)(1+V_T)}+P_T, 1+ W_T))\eqsp,
\end{equation}
where $W_T = \sum_{t\in\timesteps}\norminf{A_t - T^{-1}\sum_{\tau\in\timesteps}A_{\tau}}=\Omega(V_T)$. On the other hand, if we consider for any $t\in [T]$, $A_t = B + \parenthese{-1}^t C$ where
$$
B=\frac{1}{2}\begin{pmatrix}
    1 & 1 \\
    1 & 1
\end{pmatrix},\; C =\frac{1}{2}\begin{pmatrix}
    1 & 1 \\
    -1 & -1
\end{pmatrix} \eqsp,
$$
it is not hard to check that $$\cE_t = \begin{cases}
\defEnsLigne{(1,0),\,(\frac{1}{2},\frac{1}{2})} &\text{if $t$ is even}\\
\defEnsLigne{(0,1), (\frac{1}{2}, \frac{1}{2})}&\text{otherwise}
\end{cases}\eqsp.$$ This implies that $P_T =2T$. Likewise, one can verify that $V_T = T$, so the bound in \eqref{example1:boundzhang} grows linearly with $T$. At the same time, we remark that $Y_t = -Y_{t-1}$ with $Y_t = A_{t}-A_{t-1}$. This shows that $(A_t)_{t\in\timesteps}$ is a deterministic process (more precisely, a deterministic \emph{ARIMA(1,1,0)} process \citep{hamilton2020time}). 
\end{example}
\section{Model.}\label{section:model}
\paragraph{Notation.}In what follows, we denote the $\ell$-th coordinate of any vector $x\in\R^d$ by $x[\ell]\in\R$. Likewise, the $\ell$-th row of any matrix $\bX \in\R^{d\times K}$ is denoted by $\bX[\ell]\in\R^K$. For any vectors $(x, y)\in\R^d \times \R^d$, we write $\ps{x}{y}=x^{\top}y$ the standard euclidian inner product and $x . y = (x[1] \,y[1] , \ldots , x[d]\, y[d])^{\top}$ the Hadamard product. $\mathscr{P} (\cA)$ denotes the set of probability measures over a measurable space $\cA$, and $\Delta_{K}=\defEnsLigne{w\in\R^{K}:\:\forall \ell \in [K],\,w^j [\ell] \geq 0\;\text{and}\; \sum_{\ell=1}^{K}w^j [\ell] = 1}$ the simplex of dimension $K>0$. When $\cA = \cA^1 \times \ldots\times \cA ^J$ is the product of $J>0$ spaces, we write $\cA^{-j}=\cA^1 \times \ldots \times \cA ^{j-1}\times \cA^{j+1}\times \ldots \cA^J$ for any $j\in\agents$, so $\cA = \cA^j \times \cA^{-j}$. For any $\bw\in\mathscr{P}(\cA)$, we write $\E_{a\sim\bw}[a]=\int a\, \mathrm{d}\bw (a)$ the associated expectation. When the context is clear, we rather write $\E_{\bw}$ instead of $\E_{\ba\sim\bw}$. When $\bw=w^1 \otimes\ldots\otimes w^J$ is a product of $J>0$ measures, we define for any $j\in\agents$ $\bw^{-j}=w^1 \otimes \ldots w^{j-1}\otimes w^{j+1}\otimes \ldots \otimes w^J$ and $\E_{\bw^{-j}}$ the associated expectation operator.

\paragraph{Setting.}We consider a set of $J>0$ agents denoted by $\agents$. We suppose that each agent has access to an action set $\cA^j=\defEnsLigne{a^j _1 , \ldots , a^j _K}$ with $| \cA^j| = K$. In addition, we assume that 
the cost function of agent $j\in\agents$ is given for $Z\in\cZ\subseteq\R^d$ and  $\phi^j : \cA \rightarrow \R^d$
by: 
\begin{equation}
    \label{def:generalutility}
    c^{j} (\bw,Z) =\EEs{\ba\sim\bw}{\ps{\phi^{j}(\ba)}{Z}}\eqsp,
\end{equation}
where $\bw \in \mathscr{P}(\mca)$. Typically, we will consider $\bw= w^1 \otimes\ldots\otimes w^J $ where $w^j \in\Delta_{K}$ is a mixed strategy played by $j\in\agents$.
This cost function is flexible and is customary in contextual optimization \citep{sadana2024survey} and contextual bandit \citep{li2010contextual, lattimore2020bandit}. \eqref{def:generalutility}, $\phi^j$ represents a standard payoff function, while $Z\in\cZ$ can be interpreted as a state of nature that linearly influences preferences. Note that a time-varying game can easily be constructed by considering a sequence of states of nature $(Z_1 , \ldots, Z_T)\in\cZ^T$ for $T>0$. We rewrite \eqref{def:generalutility} in a more compact way with the following lemma. 

\begin{restatable}{lemma}{lemmarewritingutility}\label{lemma:rewritingutility}
    Let $j\in\agents$, $\bw\in\mathscr{P}(\cA)$ with $\bw=w^j \otimes \bw^{-j}$ and $\Phi^{j}(\bw^{-j})=
    (\E_{\bw^{-j}}[\phi^{j} (a^j _k, \ba^{-j})[\ell]])_{\ell , k} \in\R^{d \times K}$. We have:
    $$
c^j (\bw, Z)=\ps{Z}{\Phi^{j}(\bw^{-j})w^j}\eqsp.
    $$
\end{restatable} In 
\Cref{lemma:rewritingutility}, $\Phi^j (\bw^{-j})$ is a matrix whose column $k$ contains the cost of playing the  pure action $a^j _k$ when opponents play their mixed-strategy $\bw^{-j}$. This quantity appears naturally in online learning for games \citep{syrgkanis2015}. Note that by \Cref{lemma:rewritingutility}, $\Phi^j (\bw^{-j})$ entirely determines $c^j$, so having access to this matrix is equivalent to having access to $c^j$. Moreover, \Cref{lemma:rewritingutility} stresses that $c^j$ is conveniently linear in $w^j \in\Delta_K$ for any $j\in\agents$.

We introduce the two following assumptions for the rest of the analysis.
    \begin{assumption}
        \label{assumption:boundedutility}For any $j\in\agents$, $\ba\in\cA$ and $\cZ\in\cZ$, $\abs{\ps{Z}{\phi^j (\ba)}}\leq 1$.
\end{assumption}This boundedness assumption  is usual in learning in games and more generally in online learning \citep{hazan2014}. In particular, \Cref{assumption:boundedutility} ensures that for any $j\in\agents$, $\bw\in\mathscr{P}(\cA)$ and $Z\in\cZ$, $c^j (\bw,Z)\leq 1$.     
    \begin{assumption}
    \label{assumption:finitecontext}The set $\cZ$ is finite: $\cZ = \defEnsLigne{z_1, \ldots, z_m}$ for $m>0$.    \end{assumption}
   Assuming a finite context set is customary in bandit theory \citep{lattimore2020bandit, slivkins2019introduction} and game theory \citep{kamenica2011bayesian, kamenica2019bayesian}, and is often relevant in practical settings. We believe that the analysis to an infinite context set is possible, but such an extension falls outside the scope of the current paper, as it would require introducing fundamentally different concepts and proof techniques, see \Cref{appendix:discussionfiniteness} for further discussion.
   
    We assume that agents play a time-varying game, which is determined by a sequence of states of nature $(Z_1, \ldots, Z_T)\in\cZ^T$ of length $T>0$. At the beginning of each period $t\in\timesteps$, nature draws a state of nature $Z_t \in\cZ$, which is not revealed to agents, while  each player $j\in\agents$ receives a signal $\hZ^j _t \in \cZ$ about this state. They then select a strategy $w_t^j \in\Delta_K$  based on this signal. Finally, each agent $j$ get as a feedback the cost matrix $\Phi^j(\bw^{-j}_t)$ as well as the actual state of nature $Z_t$.  
\begin{remark}[label=remark:supervisedlearning]
    In many practical settings, the private signals $\hZ^j _t \in\cZ$ for $j\in\agents$ and $t\in\timesteps$ are predictions made by supervised learning algorithms. In this case, at the beginning of each round $t\in\timesteps$, each agent $j\in\agents$ observes covariates $X^j _t \in\cX$ and makes a prediction
    $$
\hat{Z}^j _t = g ^j _t (X^j _t)\eqsp,
    $$
    where $g^j _t \in\cG\subset \defEnsLigne{g: \cX\rightarrow\cZ}$ is some prediction algorithm based on the history of observations up to time $t$. Under \Cref{assumption:finitecontext}, this situation corresponds to multiclass online learning, for which several theoretical results are available in the litterature \citep{daniely2014, daniely2014optimal2}. 
\end{remark}
To formally describe the game, we define $\Pi^j$ as the set of policies $ \pi ^j : \parentheseLigne{\cup_{t\in\timesteps}\cH^j _t}\times\cZ \rightarrow \Delta_{K}$ for player $j\in\agents$, where $\cH^j_t$ is the set of histories at time $t\in\timesteps$ with elements $h^j_\tau =\{\Phi^{j}(\bw^{-j}_{\tau}),Z_t\}_{1\leq\tau\leq t}$. At the beginning of the game, $h^j _0 = \emptyset$. Then for any $t\in\timesteps$,
\begin{enumerate}
    \item  Each agent $j\in\agents$ observes a private signal $\hZ^j _t \in\cZ$, and picks a mixed strategy $w^j _t \in\Delta_{K}$ where $w^j _t$ is the output of a policy $\pi^j _t = \pi ^j (h^j _{t-1},\centraldot):\cZ\rightarrow\Delta_{K}$, that is $w^j _t = \pi^j _t (\hZ^j _t)$.
    \item Each agent $j$ incurs a cost $\langle Z_t , \Phi^j (\bw^{-j}_t)w^j _t \rangle$, and gets as a feedback $(Z_t, \Phi^j (\bw^{-j}_t))$. They then update $h_{t} = h_{t-1} \cup \{\Phi^j (\bw^{-j} _t) , Z_t\}$. 
\end{enumerate}
\begin{remark}[continues=remark:supervisedlearning]
In the case where private signals are predictions from an online algorithm, agents train policies $\kappa^j : \cX\rightarrow\Delta_{K}$ mapping covariates to strategies. Indeed, for any $j\in\agents$ and $t\in\timesteps$:
$$
w^j _t = \pi^j _t (\hZ^j _t) = (\pi^j _t \circ g^j_t) (X^j _t)=\kappa^j _t (X^j_t)\eqsp.
$$
\end{remark}
We consider the standard full-information feedback setting, where each player \( j \in \agents \) observes \( \Phi^j (\bw^j _t) \). We believe that extending our results to bandit feedback -- i.e., when agents only observe the reward from their realized action-- \citep{foster2016} is feasible, though it would require additional technical refinements.
\paragraph{Regrets.}We now present the two regret concepts used in this paper to quantify the optimality of a policy $\pi^j \in\Pi^j$. They are essentially contextual versions of the classic external \citep{zinkevich2003online} and swap \citep{blum2007external} regrets. In what follows, $\scrT^z = \defEnsLigne{t\in\timesteps:\: Z_t =z}$ denotes the timesteps at which $z\in\cZ$. 

 First, following \citet{sessa2021}, we introduce a contextual external regret. For any $j\in\agents$, given a fixed sequence of opponent strategies $(\bw^{-j}_t)_{t\in\timesteps}$, we denote by  $\pi^j _\star :\cZ\rightarrow \Delta_K$ the static comparator which satisfies
    \begin{equation}
        \label{def:staticomparator}
        \sum_{t\in\scrT^z}c^j (\pi^j _{\star} (z),\bw^{-j}_t, Z_t) \leq  \sum_{t\in\scrT^z}c^j (w,\bw^{-j}_t , Z_t)\eqsp,
    \end{equation}
     for any $z\in\cZ$ and $w\in\mathcal{P}(\cA^j)$. The comparator $\pi^j_\star$ maps each context $z$ to the best action in hindsight on the time steps when $z$ was observed. Denoting $w^j _t = \pi ^j _t (\hZ^j _t)$ the strategy of agent $j$ for any $t\in\timesteps$, we then define the following contextual external regret:
    \begin{equation}
\label{eq:regretlinearized2}\regret_T ^j =\sum_{t\in\timesteps}\parentheseDeux{c^j (w^j_ t, \bw^{-j}_t, Z_t)-c^j (\pi^j _\star (Z _t),\bw^{-j}_t, Z_t)}\eqsp.
    \end{equation}Note that $\regret^j _T$ is not fully static as the comparator is allowed to vary from a context to another. In this sense, \eqref{eq:regretlinearized2} can be viewed as an intermediary between external regret and dynamic regret \citep{hall2013dynamical, besbes2015non}. The existing literature on time-varying games typically focuses on dynamic regret \citep{duvocelle2018learning, zhang2022no, anagnostides2024convergence}, which is the most stringent notion of regret. However, the resulting bounds often include a path length term. This term captures the intrinsic variation of the comparator sequence, such as $P_T$ and $V_T$ in \Cref{example:variation}. In contrast, we will show that bounds on regret \eqref{eq:regretlinearized2} depend only on a prediction error term, which vanishes if agents are able to accurately predict the states of nature.

    While external regret has been a cornerstone measure of performance in online learning for games, swap regret has recently aroused a lot of interest since it sets a more demanding learning benchmark and leads to tighter equilibrium concepts \citep{anagnostides2022uncoupled, peng2024fast, dagan2024external}. We therefore supplement our analysis of regret \Cref{eq:regretlinearized2} by a similar study of contextual swap regret, which we introduce below.  Let $\Lambda=\defEnsLigne{\lambda: \Delta_K \rightarrow \Delta_K}$ be the set of swap deviation maps, and define for any $j\in\agents$, $\lambda^{j} _{\star}:\Delta_K \times \cZ \rightarrow \Delta_K$ the optimal swap comparator, which satisfies for any $z\in\cZ$ and any $\lambda\in\Lambda$:
    $$
\sum_{t\in\scrT^z}c^j (\lambda^j _\star (w^j _t , z), \bw^{-j}_t , z)\leq \sum_{t\in\scrT^z}c^j (\lambda(w^j _t), \bw^{-j}_t , z)\eqsp. $$
Given a context $z$, the swap comparator $\lambda^j_\star$ maps any played strategy $w^j_t$ to an alternative strategy that would have achieved better performance on the time steps when $z$ was observed. Importantly, competing against $\lambda^j_\star$ is strictly more demanding than competing against $\pi^j_\star$ from \eqref{def:staticomparator}. While $\pi^j_\star$ assigns a fixed optimal action to each context $z$, independent of the strategies actually played, $\lambda^j_\star(w, z)$ adapts to the specific strategy $w$ used. This flexibility allows $\lambda^j_\star$ to identify better-performing alternatives conditioned on past decisions, making it a stronger—and thus more challenging—benchmark.

With $w^j _t = \pi^j _t (\hZ^j _t)$ being the strategy played by agent $j$ at time $t$, we define the following contextual swap regret:
    \begin{equation}
        \label{def:swapregret}
        \uregret^j _T = \sum_{t\in\timesteps}\parentheseDeux{c^j (w^j _t, \bw^{-j}, Z_t) -c^j (\lambda^j _{\star}(w^j _t , Z_t),\bw^{-j}_t , Z_t)}\eqsp.
    \end{equation}
    
    Finally, in the non-contextual case, the Blum-Mansour reduction \citep{blum2007external} is a convenient procedure which allows to design a no-swap regret algorithm from any no-external regret one. A natural question is whether a similar reduction exists in our setting, that is, whether an algorithm minimizing \eqref{def:swapregret} can be obtained from an algorithm minimizing \eqref{eq:regretlinearized2}. We answer by the positive with the following proposition.
    \begin{restatable}{proposition}{externaltoswapregret}\label{proposition:externaltoswapregret}
        Assume that player $j\in\agents$ plays an algorithm $\pi^j \in\Pi^j$ achieving $\regret^j_ T \leq f(J,T,K,m)$ for some $f: \N^4 _{+}\rightarrow \R_{+}$. Then, one can design an algorithm $\overline{\pi}^j \in\Pi^j$ achieving
        $$
\uregret^j _T \leq K f(J, T,K,m)\eqsp.
        $$
    \end{restatable}The explicit procedure to design $\bar{\pi}^j $ from $\pi^j$ is described in \Cref{algorithm:blummansour}, and the proof of \Cref{proposition:externaltoswapregret} is deferred to \Cref{appendix:proofs}. A direct consequence of \Cref{proposition:externaltoswapregret} is that any algorithm with a guarantee on external regret \eqref{eq:regretlinearized2} can be converted into another algorithm with a guarantee on swap regret \eqref{def:swapregret}, at the cost of an additional $K$ factor. This will be particularly useful to extend the analysis from external to swap regret. Note that more recent procedures \citep{dagan2024external, peng2023fastswapregretminimization} allow to deal with larger action spaces by reducing the dependence of $K$, yet at the cost of a degraded dependence on $T$.

\paragraph{Equilibrium.}We consider two equilibrium concepts, which naturally relates to the two regrets previously defined. First, we focus on the classic contextual coarse-correlated equilibrium \citep{sessa2021, maddux2024}, whose definition is recalled below. 
\begin{restatable}{definition}{definitioncoarsecorrelatedequilibrium}[\citealt{sessa2021}]\label{definition:coarsecorrelatedeq}
    Let $\varepsilon>0$. An \emph{$\varepsilon$-contextual coarse-correlated equilibrium} is a joint policy $\bldnu: \cZ\rightarrow\mathscr{P}(\cA)$ such that for any $j\in\agents$ and $\pi^j \in\Pi^j$:
    \begin{align*}
        T^{-1}&\sum_{t\in\timesteps}c^j (\nu^j (Z_t), \bldnu^{-j}(Z_t), Z_t) \\&\leq T^{-1}\sum_{t\in\timesteps}c^j (\pi^j (Z_t), \bldnu^{-j}(Z_t), Z_t)+\varepsilon\eqsp.
    \end{align*}
\end{restatable}
 Note that \Cref{definition:coarsecorrelatedeq} extends the classic coarse correlated equilibrium concept \citep{foster1998asymptotic} to the case where the underlying state of nature changes over time. A more detailed interpretation of \Cref{definition:coarsecorrelatedeq} is provided in \Cref{appendix:equilibrium}.

While coarse-correlated equilibrium has been extensively studied, it is arguably weak in the sense that it only prevents coarse deviations. On the other hand, correlated equilibrium \citep{aumann1987correlated} is a tighter equilibrium concept which prevents swap deviations, see \Cref{appendix:equilibrium} for more discussion. We introduce below an equivalent concept adapted to our framework. In the following definition, $\varrho^j : \cA^j \times \cZ \rightarrow \cA^j$ is any swap deviation function, which given an action and context $(a^j, z)$ returns an alternative action $\tilde{a}^{j}$.
\begin{restatable}{definition}{definitioncorrelateqauilibrium}
    \label{definition:correlatedeq}Let $\varepsilon > 0$. An \emph{$\varepsilon$-contextual correlated equilibrium} is a joint policy $\overline{\bldnu}:\cZ\rightarrow\mathscr{P}(\cA)$ such that for any $j\in\agents$ and $\varrho^j :\: \cA^j \times \cZ \rightarrow \cA^j$:
    \begin{align*}
T^{-1}&\sum_{t\in\timesteps}\EEs{\ba\sim\overline{\bldnu}(Z_t)}{\ps{\phi^j (\ba)}{Z_t}} \\ &\leq T^{-1}\sum_{t\in\timesteps}\EEs{\ba\sim\overline{\bldnu}(Z_t)}{\ps{\phi^j (\varrho^j (a^j ,Z_t), \ba^{-j})}{Z_t}} + \varepsilon\eqsp.
    \end{align*}
\end{restatable}\Cref{definition:correlatedeq} extends the classic correlated equilibrium notion \citep{aumann1987correlated} to the contextual case, by letting the swap functions $\varrho^j (\centraldot, z)$ depend on the state of nature.

In the non-contextual case, there exists a well-known and powerful connection between external regret minimization and coarse-correlated equilibrium. If all players use no-external regret algorithms, their empirical average strategy is an approximate coarse correlated equilibrium. Crucially, we show that the same property holds in our contextual framework. In the following proposition, $n_z = \abs{\scrT^z}$ denotes the number of time state $z\in\cZ$ occurs.
\begin{restatable}{proposition}{epscoarsecorrelatedeq}\label{proposition:noregretcoarsecorrelatedequilibrium}
    Assume that for any $j\in\agents$, agent $j$ uses a policy $\pi^j \in\Pi^j$ incurring an external regret $\regret^j _T$ as in \eqref{eq:regretlinearized2}, and denote by $w^j _t = \pi^j _t (\hZ^j _t)$ for any $t\in\timesteps$. Let $\hat{\bldnu}_T : \cZ \rightarrow \mathscr{P}(\cA)$ be such that for any $z\in\cZ$, 
    $$
    \hat{\bldnu}_T (z) = \begin{cases}
        n_z^{-1}\sum_{t\in\scrT^z}w^1 _t \otimes \ldots \otimes w^J _t &\text{if }n_z > 0 \eqsp, \\
        (K^{-1} ,\ldots, K^{-1}) &\text{otherwise}\eqsp.
    \end{cases}
    $$Then, $\hat{\bldnu}_T$ is an $\varepsilon$-contextual coarse correlated equilibrium with $$\vareps = \max_{j\in\agents}T^{-1}\regret^j _T\eqsp.$$
\end{restatable}The proof of \Cref{proposition:noregretcoarsecorrelatedequilibrium} is deferred to \Cref{appendix:proofs}. It is clear from this proposition that if $\regret^j _T =o(T)$ for every $j\in\agents$, $\boldsymbol{\hat{\nu}}_T$ asymptotically convergences to an exact coarse-correlated equilibrium. 

Moreover, the same connection exists for swap regret and correlated equilibrium in the non-contextual setting. We also retrieve this property in our contextual setting, as showed by the following proposition.

\begin{restatable}{proposition}{epscorrelatedeq}\label{proposition:noregretcorrelatedequilibrium}
    Assume that for any $j\in\agents$, agent $j$ uses a policy $\bar{\pi}^j \in\Pi^j$ incurring a swap regret $\bar{\regret}^j _T$ defined as in \eqref{def:swapregret}. Let $\hat{\bldnu}_T : \cZ \rightarrow \mathscr{P}(\cA)$ be defined as in \Cref{definition:coarsecorrelatedeq}. Then, $\hat{\bldnu}_T$ is an $\varepsilon$-contextual correlated equilibrium with $$\vareps = \max_{j\in\agents}T^{-1}\bar{\regret}^j _T\eqsp.$$
\end{restatable}The proof of this result can be found in \Cref{appendix:proofs}. It shows that if all players use algorithms incurring a sub-linear swap regret, their empirical average strategy converges to an exact correlated equilibrium. Both \Cref{proposition:noregretcoarsecorrelatedequilibrium} and \Cref{proposition:noregretcorrelatedequilibrium} are key to our analysis, since they convert individual regret guarantees into convergence rates to equilibrium. Hence, bounding individual regrets is our first objective.

\paragraph{Social welfare.}On top of convergence to equilibrium, we study social welfare, and in particular whether no-regret strategies may result in a welfare close to the optimal one. In non-contextual games,  the so-called Roughgarden smoothness condition \citep{roughgarden2015intrinsic} is particularly convenient to address this question \citep{syrgkanis2015}. This condition---which is satisfied by a wide class of games, including congestion games \citep{roughgarden2002bad, christodoulou2005price}, facility games and second price auctions \citep{roughgarden2015intrinsic}---states that even when players deviate from the optimal strategy, the total cost in a game doesn't increase too much. Under this condition, it can be shown that the average social cost converges to the optimal one times the price of anarchy. 

Here, we assume that our game satisfies the contextual counterpart to the Roughgarden smoothness condition.
\begin{assumption}
    \label{assumption:smoothness} There exist $\delta>0$ and $\mu>0$ such that for any $\ba\in\cA$,  $\ba_\star \in\cA$ and $z\in\cZ$,
    \begin{equation*}
        \sum_{j\in\agents}\ps{z}{\phi_{j}(a^j _\star , \ba_{-j})}\leq\sum_{j\in\agents}\parentheseDeux{\delta \ps{z}{\phi_{j}(\ba_ \star)} +\mu \ps{z}{\phi_{j}(\ba)}}\eqsp.
    \end{equation*}
\end{assumption}It is well known that under \Cref{assumption:smoothness},  $\gamma = \delta(1-\mu)^{-1}$ is an upper bound on the price of anarchy \citep{roughgarden2015intrinsic}. In what follows,
$$C_t (\bw_t)=\sum_{j\in\agents}c^j (\bw_t , Z_t)\eqsp,$$ denotes the social cost at time $t\in\timesteps$ and $$C^\star = \min_{\boldsymbol{\rho} : \cZ\rightarrow \cA}T^{-1}\sum_{t\in\timesteps}\sum_{j\in\agents}c^j (\boldsymbol{\rho}(Z_t), Z_t)\eqsp,$$ the optimal average social cost in pure strategy. The following proposition shows that under \Cref{assumption:smoothness}, the distance between the average social cost and the optimal one is bounded by the sum of external contextual regrets.
\begin{restatable}{proposition}{propboundaveragesocialwelfare}\label{proposition:averagewelfareoptimalwelfare}
    Assume \Cref{assumption:smoothness}. Then with $\gamma = \delta(1-\mu)^{-1}$, 
    $$
\frac{1}{T}\sum_{t\in\timesteps}C_t(\bw_t) \leq \gamma C^\star +\frac{1}{(1-\mu)T}\sum_{j\in\agents}\regret^j _T\eqsp.
    $$
\end{restatable}The proof of \Cref{proposition:averagewelfareoptimalwelfare} can be found in \Cref{appendix:proofs}. In particular, when $\sum_{j\in\agents}\regret^j _T = o (T)$, the average social cost is guaranteed to converge to a fraction of the optimal one. Therefore, bounding $\sum_{j\in\agents}\regret^j _T$ will be our second objective.  
\section{Prediction-aware learning.}\label{section:predictionaware}\paragraph{Algorithm.}In the non-contextual case, the optimistic Multiplicative Weight Update (\texttt{OMWU}) algorithm has proven particularly effective for controlling individual and social regrets in uncoupled multiplayer games. We propose below the predictive-\texttt{OMWU} algorithm, abbreviated \texttt{POMWU}, which is an extension of \texttt{OMWU} to our framework. Broadly speaking, \texttt{POMWU} maintains one \texttt{OMWU} instance per context. At the beginning of each round, agents predict the context and use the corresponding \texttt{OMWU} to play. Once the actual state of nature has been revealed, they update the algorithm based on the cost feedback for future rounds. The pseudo-code of \texttt{POMWU} is displayed in \Cref{algorithm:ohedge}. 
\begin{algorithm}[!ht]
\caption{Optimistic MWU with predicted contexts (\texttt{POMWU}) for agent $j\in\agents$.}
\label{algorithm:ohedge}
\begin{algorithmic}[1]
\STATE Initialize $\rho_{z_{1}}=\ldots =\rho_{z_{m}}= (K^{-1},\ldots, K^{-1})$ and $\Psi_{z_1} = \ldots = \Psi_{z_m}=\boldsymbol{0}_{d\times K}$.
\FOR{each $t \in\timesteps$}
\STATE Predict $\hZ^j_t \in\cZ$, set $M^j _t = \Psi_{\hZ^j _t}$ and $g^j_t = \rho_{\hZ ^j _t}$.
\STATE Play $w^j _t \in\Delta_K$ where for each $\ell\in\iint{1}{K}$,
$$
w^j _ t[\ell] = \frac{g^j_t [\ell] \exp(-\eta M^{j}_t [\ell] \hZ^j _t)}{\sum_{k\in[K]}g^j _t [k] \exp(-\eta M^{j}_t [k] \hZ^j _t )}
$$
\STATE Observe $Z_t \in\R^d$ and $\Phi^{j}(\bw^{-j}_t)$.
\STATE Update $\Psi_{Z_t}\leftarrow \Phi^{j}(\bw^{-j}_{t})$
\STATE Update $\rho_{Z_t} \leftarrow \rho_{Z_t}.\exp (-\eta \Phi^{j} (\bw^{-j}_t)^{\top}Z_t)\eqsp.$
\ENDFOR
\end{algorithmic}
\end{algorithm}
\paragraph{RVU analysis.}The \textit{Regret bounded by Variation in Utility} (RVU) bound \citep{syrgkanis2015} is an effective method for deriving guarantees regarding individual regrets in multiplayer games. Intuitively, this approach ties players' utilities to the variation in observed utilities and played strategies. The RVU ensures that players' regrets remain low when they  efficiently adapt their strategy to counter the variation in payoffs. In the following key lemma, we establish a \textit{contextual} RVU bound adapted to our framework. In what follows, we write $\scrT^z = \defEnsLigne{t^z _1 , \ldots, t^z _{n_z}}$ for any $z\in\cZ$, and  $L^j _T = \sum_{t\in\timesteps}\2{\hZ^j _t \ne Z_t}$ the total number of mis-predictions made by agent $j\in\agents$ throughout of the game. 
\begin{restatable}{proposition}{propcontextualrvu}\label{prop:contextualrvu}
        Assume \Cref{assumption:boundedutility} and \Cref{assumption:finitecontext}
. Any $j\in\agents$ applying \Cref{algorithm:ohedge} with learning rate $\eta > 0$ has an external regret bounded as follows:
        \begin{align*}
                  &\regret^j _T \leq  \frac{(5+\ln (K))L^j _T + m\ln (K)}{\eta} \\
       &+ \eta\parenthese{\sum_{z\in\cZ}\sum_{i\leq n_z}\norminf{\parenthese{\Phi^{j}(\bw^{-j}_{t^z _i}) - \Phi^{j}(\bw^{-j}_{t^z _{i-1}})}^{\top}z}^2 + 4L^j _T} \\
       &- \frac{1}{16\eta}\sum_{z\in\cZ}\sum_{i\leq n_z}\normone{w^{j}_{t^z _i} - w^{j}_{t^z _{i-1}}}^2\eqsp. 
        \end{align*}
\end{restatable}Contrary to the classic RVU approach \citep{syrgkanis2015}, the bound in \Cref{prop:contextualrvu} depends on the lengths of the \textit{context-specific} paths $\Phi^j (\bw^{-j}_{t^z _1}),\ldots,\Phi^j (\bw^{-j}_{t^z _{n_z}})$ and $w^j _{t^z _1},\ldots, w^j _{t^z _{n_z}}$. The need for this new contextual RVU stems from the fact that players may mispredict states of nature at different periods, preventing the naive use of a classic RVU, see \Cref{appendix:proofs} for more details. Note that in its current form, \Cref{prop:contextualrvu} holds for any arbitrary sequence of strategies by other agents, and does not provide an explicit bound for individual regrets.
\begin{remark}[continues=remark:supervisedlearning]
It is possible to quantify $L^j _T$ under \Cref{assumption:finitecontext} when agents use an online algorithm for predicting $(Z_t)_{t\in\timesteps}$. Indeed, this boils down to multiclass online classification problem, for which bounds on $L^j _T$ have been established by \citet{daniely2014}. Assume that $\cG$ has a finite Littlestone dimension $\text{dim}_{\mathscr{L}}(\cG)<\infty$ \citep{littlestone1988}. In the realizable case, that is when for every $j\in\agents$, there exists $g_j ^\star \in\cG$ such that $Z_t = g_j^\star (X^j _t)$ for any $t\in\timesteps$, there exists an online algorithm $g^j _t : \cX\rightarrow \cZ$ such that $L^j _T = \sum_{t\in\timesteps}\2{g^j_t (X^j_t) \ne Z_t}$ satisfies:
\begin{equation}
    \label{eq:danielyone}
    L^j _t \leq \text{dim}_{\mathscr{L}}(\cG)\eqsp.
\end{equation}
In the agnostic case, denoting $L^{\star j}_T = \min_{g_j\in\cG}\sum_{t=1}^{T}\2{g_j (X^j _t)\ne Z_t}$, there exists an algorithm such that 
\begin{equation}
    \label{eq:danielytwo}
    L^j _T \leq L^{\star j}_T + \sqrt{\frac{1}{2}\text{dim}_{\mathscr{L}}(\cG)T\ln (Tm)}\eqsp.
\end{equation}
The algorithms leading to \eqref{eq:danielyone} and \eqref{eq:danielytwo}, namely \Cref{algorithm:daniely1} and \Cref{algorithm:daniely2}, are both recalled in \Cref{appendix:algos}.
\end{remark}
\paragraph{Equilibrium.}Equipped with \Cref{prop:contextualrvu}, we first focus on convergence to equilibrium. As discussed in \Cref{section:model}, this only requires bounding individual regrets. The following proposition, which follows from \Cref{prop:contextualrvu}, establishes this bound.

\begin{restatable}{proposition}{propconvergenceindividualutility}\label{prop:boundindividualregret}
    Define $\overline{L}_T = \max_{j\in\agents}L^j _T$ and assume \Cref{assumption:boundedutility} and   \Cref{assumption:finitecontext}. If all agents use \Cref{algorithm:ohedge} with a learning rate $\eta > 0$, then for any $j\in\agents$:
    \begin{align*}
        \regret^j _T &\leq\frac{(5+\ln (K))\overline{L}_T + m\ln (K)}{\eta} \\
        &\quad+ \eta \parentheseDeux{(J-1)^2 (9T\eta^2 + 4\overline{L}_T) + 4 \overline{L}_T }\eqsp.
    \end{align*}
    In particular if $T=\Omega (J^2 \overline{L}_T)$, setting $\eta^\star =\Theta( J^{-1/2}T^{-1/4}[\ln (K) (\overline{L}_T +m)]^{1/4})$ leads to:
    $$
\regret^j _T = \bigO (\,[\ln (K) (\overline{L}_T + m) ] ^{3/4} T^{1/4}J^{1/2}\,)\eqsp.
    $$
\end{restatable}
In the realizable case of \Cref{remark:supervisedlearning} where $\overline{L}_T = \bigO_{T}(1)$, we recover the result $\regret^j _T = \bigO (T^{1/4}J^{1/2})$ from \citet{syrgkanis2015}. We also observe that setting the learning rate to $\eta^\star$ requires agents to know $\overline{L}^j _T$. This is reasonable if they use the same hypothesis class, since uniform bounds on $\overline{L}_T$ are known (see e.g., \Cref{remark:supervisedlearning}). 
\begin{remark}[continues=remark:supervisedlearning]
Recently, collaborative and federated learning has emerged as a topic of prime importance in Machine learning \citep{blum2017collaborative, kairouz2021advances}. One may wonder whether agents sharing a common model, so $\hZ^j _t = \hZ_t \in\cZ$ for any $j\in\agents$, may improve \Cref{prop:boundindividualregret}. Indeed, even though agents play uncoupled strategies, policies $\pi^j _t (\hZ_t)$ are implicitly coordinated as they rely on a same signal. We show in \Cref{prop:imrovementregretbound} in \Cref{appendix:proofs} that in this case, we can drop the assumption $T=\Omega(J^2 \overline{L}_T)$ and still recover the guarantee of \Cref{prop:boundindividualregret} by a direct improvement of the proof. Studying the  impacts of collaborative learning in games more broadly is an interesting topic for future research. 
\end{remark}

It is now possible to obtain an explicit convergence rate to coarse-correlated equilibrium by combining \Cref{proposition:noregretcoarsecorrelatedequilibrium} with \Cref{prop:boundindividualregret}.
\begin{corollary}
    \label{corollary:convergenceeq}
    Assume \Cref{assumption:boundedutility},  \Cref{assumption:finitecontext} and $T=\Omega (J^2 \overline{L}_T)$. If all agents use \Cref{algorithm:ohedge} with $\eta^\star >0$ as defined in \Cref{prop:boundindividualregret}, then $\hat{\bldnu}_T$ (as defined as in \Cref{proposition:noregretcoarsecorrelatedequilibrium}) is an $\varepsilon$-coarse correlated equilibrium, with $$\varepsilon =  \bigO (\,[\ln (K) (\overline{L}_T + m) ] ^{3/4} T^{-3/4}J^{1/2}\,)\eqsp.$$
\end{corollary}In addition, it is possible to obtain a convergence rate for the more demanding contextual correlated equilibrium concept defined in \Cref{definition:correlatedeq}. As a matter of fact, the reduction described in \Cref{algorithm:blummansour} allows to transform \texttt{POMWU}, which enjoys the guarantee on $\regret^j _T$ presented in \Cref{prop:boundindividualregret}, into an algorithm $\bar{\pi}^j$ with a guarantee on $\uregret^j _T$---see \Cref{proposition:externaltoswapregret}. Then, applying \Cref{proposition:noregretcorrelatedequilibrium} leads to the following corollary.
\begin{corollary}    \label{corollary:convergenceeq}
    Assume \Cref{assumption:boundedutility},  \Cref{assumption:finitecontext} and $T=\Omega (J^2 \overline{L}_T)$. If all agents use \Cref{algorithm:ohedge} in conjonction with \Cref{algorithm:blummansour} and  $\eta^\star >0$ as in \Cref{prop:boundindividualregret}, then $\hat{\bldnu}_T$  is an $\bar{\varepsilon}$-correlated equilibrium, with $$\bar{\varepsilon} =  \bigO (\,[K\ln (K) (\overline{L}_T + m) ] ^{3/4} T^{-3/4}J^{1/2}\,)\eqsp.$$
\end{corollary}
\paragraph{Social welfare.}
We now turn our attention to social welfare when agents use \texttt{POMWU}. As discussed in \Cref{section:model}, this requires bounding the sum of regrets. A first, naive approach would be to sum the bound of individual regrets obtained in \Cref{prop:boundindividualregret}. However, we show below that another choice of $\eta$  leads to a much better guarantee.

    \begin{restatable}{proposition}{propboundsumregret}\label{prop:boundsumregret}Let $L_T = \sum_{j\in\agents}L^j _T$, and assume \Cref{assumption:boundedutility}, \Cref{assumption:finitecontext}. 
        If all agents use \Cref{algorithm:ohedge} with a learning rate $\eta = (4(J-1))^{-1}$, then
        \begin{align*}
            \sum_{j\in\agents}\regret^j _T &\leq 4J\parentheseDeux{(5+\ln (K))L_T + mJ\ln (K)} + \frac{L_T}{J-1} \\
            &= \bigO (J \ln (K) (L_T + mJ))\eqsp.
        \end{align*}
    \end{restatable}
    Note that in the setting of \Cref{remark:supervisedlearning} under the realizable assumption, $L_T = \bigO_{T}(1)$ and hence we recover the classic result $\sum_{j\in\agents}\regret^j _T = \bigO_{T} (1)$ of \citet{syrgkanis2015} in the static setting. 
    
    The bound in \Cref{prop:boundsumregret} can immediately be converted into a convergence rate of social cost to a fraction of the optimal one via \Cref{proposition:averagewelfareoptimalwelfare}.
    \begin{corollary}\label{corollary:convergencewelfare}Assume \Cref{assumption:boundedutility}, \Cref{assumption:finitecontext} and \Cref{assumption:smoothness}. If 
Assume all agents use \Cref{algorithm:ohedge} with $\eta = (4(J-1))^{-1}$, then 
$$
\frac{1}{T}\sum_{t\in\timesteps}C_t (\bw_t) \leq \gamma C^\star + \bigO(J\ln (K) T^{-1}(L_T + mJ))\eqsp.
$$
    \end{corollary}
\paragraph{Robustness.}Finally, we turn our attention to the adversarial regime where not all agents use \texttt{POMWU}. Specifically, we ask whether the regret of \texttt{POMWU} remains low against any arbitrary sequence of cost feedback. This robustness property is a common desiderata in the literature \citep{syrgkanis2015, foster2016}.

\begin{restatable}{proposition}{robustness}\label{proposition:robustness}
    Assume \Cref{assumption:boundedutility} and \Cref{assumption:finitecontext}. If player $j\in\agents$ uses \Cref{algorithm:ohedge} with $\eta = \Theta ([\ln(K)(L^j _T + m)]^{1/2}(L^j _T +T)^{-1/2})$, then for any sequence $(\bw^{-j}_1 , \ldots, \bw^{-j}_T) \in\mathscr{P}(\cA^{-j})^{T}$:
    $$
    \regret^j _T = \bigO \parenthese{\sqrt{\ln(K)(L _T ^{j } + m)(L^j _T +T)}}\eqsp.
    $$
\end{restatable}Here again, in the setting of \Cref{remark:supervisedlearning} under realizability, $L^j _T = \bigO_T (1)$ and therefore we recover the guarantee $\regret^j _T = \bigO (\sqrt{T})$.
\section{Experiments.}\label{section:experiments}
\begin{figure*}\centering
\includegraphics[width=\textwidth]{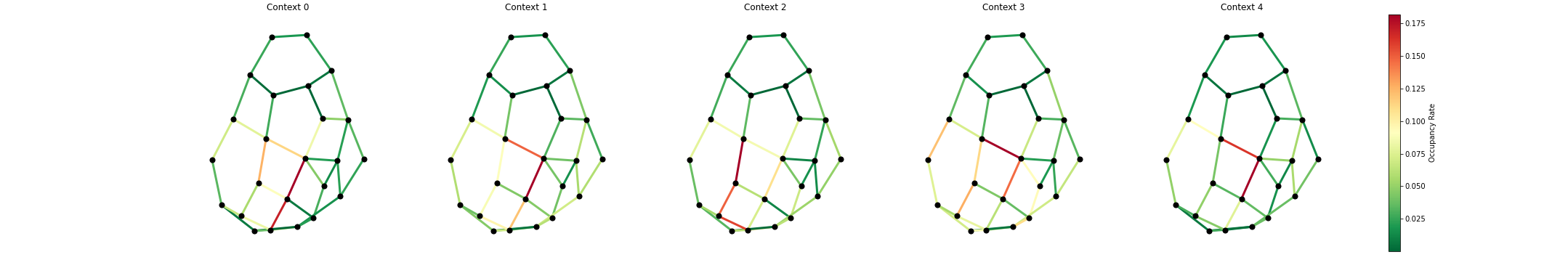}
  \caption{Average repartition of agents on the network for each context under a $10^{-3}$-coarse correlated equilibrium.}\label{fig:equilibrium}
\end{figure*}
\begin{figure*}
    \centering
    \begin{minipage}{0.32\textwidth}
        \centering
        \includegraphics[width=\linewidth]{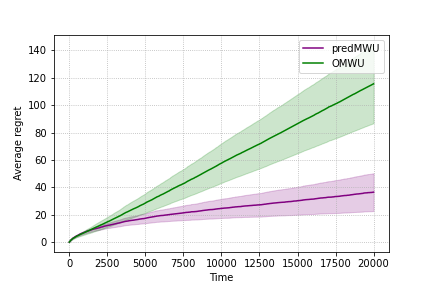}
        \caption{Average regret over agents for \texttt{POMWU} and \texttt{OMWU}.}
        \label{fig:regret}
    \end{minipage}
    \hfill
    \begin{minipage}{0.32\textwidth}
        \centering
        \includegraphics[width=\linewidth]{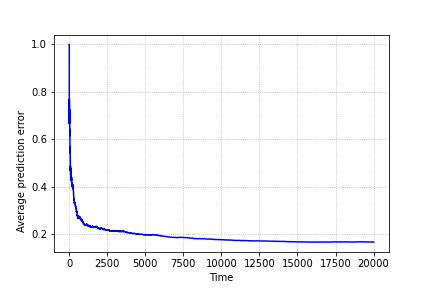}
        \caption{Average prediction error from the online logistic regression.}
        \label{fig:error}
    \end{minipage}
    \hfill
    \begin{minipage}{0.32\textwidth}
        \centering
        \includegraphics[width=\linewidth]{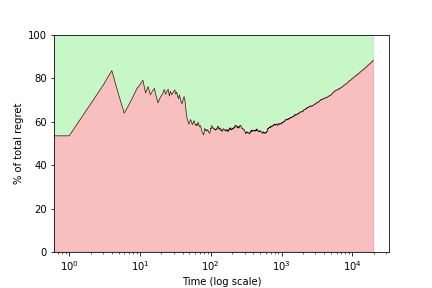}
        \caption{Proportion of average regret incurred under mispredicted contexts.}
        \label{fig:proportion}
    \end{minipage}
\end{figure*}
\paragraph{Setting.}We illustrate the performances of \texttt{POMWU} on the Sioux Falls routing problem from \citet{leblanc1975efficient} with the parameters from \citet{sessa2019no}. We consider a network of cities connected by roads. In each city, there is one agent willing to send a given quantity of goods to each other city. Agents want to minimize their travel time, which is determined by both congestion on the network, and external factors such as weather and road condition. Formally, we consider a graph $(\cV, \cE)$ with $J=\abs{\cV}(\abs{\cV}-1)$ agents, each of whom wants to send $q_j > 0$ units from $n_j \in\cV$ to $m_j \in\cV$. For any $j\in\agents$, we let $\cA^j $ be the set of $K>0$ shortest paths connecting $n_j$ to $m_j$, that is any $a^j \in\cA^j$ can be written as $a^j = (i^j _1, \ldots, i^j _R)$ with $i^j _1 = n_j$, $i^j _R = m_j$, and $(i_r, i_{r+1})\in\cE$ for any $r\in\iint{1}{R-1}$. For any profile of actions $\ba = (a^1 , \ldots, a^J)\in\cA$ and pair of nodes $(p,\ell)\in\cE$, we denote by $$\phi^j _{p,\ell}(\ba)=\begin{cases}\sum_{i\in\agents}\2{(p,\ell)\in\a^i }q_i ^4 &\text{if }(p,\ell)\in a^j  \\
0 &\text{otherwise}\eqsp,\end{cases}$$ the total congestion\footnote{In \citet{sessa2019no}, the congestion is of form $\tilde{\phi}_{p,\ell}(\ba)=(\sum_{k\in\agents}\2{(p,\ell)\in a^k }q_k)^{4}$. We only keep the term $q_k ^4$ in this sum so $\phi_{k,q}$ is linear in $\ba\in\cA$, which is necessary to compute expectations given the size of action space $\abs{\cA}=m^{\abs{\cV}(\abs{\cV}-1)}$.} faced by $j\in\agents$ on $(p,\ell)$, and $\phi^j(\ba)=(\phi^j  _{p,\ell}(\ba))_{p,\ell}\in\R^{\abs{\cV}\times\abs{\cV}}$ the corresponding matrix. Agents are allowed to randomize over routes, so they play $w^j \in\Delta_K$. To each pair $(p,\ell)\in\cV \times\cV$, we also associate a cost coefficient $z_{p,\ell}>0$ related to road condition or weather, and we denote by $Z=(z_{p,\ell})_{p,\ell}\in\R^{\abs{\cV}\times\abs{\cV}}$ the corresponding matrix. Then for any $\bw\in\mathscr{P} (\cA)$ and $Z\in\R^{\abs{\cV}\times\abs{\cV}}$, the cost for any $j\in\agents$ is given by:
$$
c^j (\bw, Z) = \EEs{\bw}{\ps{Z}{\phi^j  (\ba)}_{\text{F}}}\eqsp,
$$where $\ps{A}{B}_{\text{F}}=\trace(A^{\top}B)=\sum_{i,j}A_{i,j}B_{i,j}$ is the Frobenius inner product. $c^j $ captures the expected travel time of player $j\in\agents$ when they pick routes according to $w^j \in\Delta_K$ and other agents according to $\bw^{-j}\in\mathscr{P}(\cA^{-j})$ under context $Z\in\cZ$. Additional experimental details can be found in \Cref{appendix:experiment}.
\paragraph{Supervised learning.}In our experiment, there are $m>0$ random contexts denoted $\cZ=\defEnsLigne{z_1, \ldots,z_m}$. For any $z\in\cZ$, there exists $\beta^\star _z \in\R^b$ such that 
$$
\P (Z=z | X^0) = \zeta(\beta^\star _z , X^0)=\frac{\exp(\beta^\star _z X^0)}{\sum_{z'\in\cZ} \exp(\beta^\star _{z'}X^0)}\eqsp,
$$where $X^0 \in\R^b$ is a vector of covariates (which can be thought of as a meteorogical or a traffic forecast) drawn from a standard Normal multivariate distribution. At each round $t\in\timesteps$, agents observe $X^0 _t \in\R^b $ and predict with a logistic regression $\hZ_t \in\cZ$, that is  $\hZ_t = \argmax_{z\in\cZ}
\zeta(\hat{\beta}_z , X^0 _t)$. They then update $\hat{\beta}_{z_1}, \ldots, \hat{\beta}_{z_m}$ in an online fashion with a stochastic gradient descent. More details can be found in \Cref{appendix:experiment}.

\paragraph{Game.}There are $T>0$ rounds. At each $t\in\timesteps$, A pair $(X^0 _t ,Z_t)$ is drawn, each agent $j\in\agents$ observe $X^0 _t$ , predict $\hZ^j _t$, and play $w^j _t \in\Delta_K$ according to \Cref{algorithm:ohedge}. They then receive $Z_t$ and $(\E_{\bw^{-j}_t}[\phi^j (a^j _{t,k} , \ba^{-j}_t)])_{k\in[K]}$ as a feedback, which they use to update \texttt{POMWU} and their logistic regression. The parameters used in our experiment are summarized in \Cref{appendix:experiment}. 
\paragraph{Results.}\Cref{fig:regret} displays the the regret averaged over players\footnote{Shaded areas correspond to standard error computed over multiple runs.} for a naive \texttt{OMWU} algorithm which ignores states of nature, and \texttt{POMWU}. 
The effectiveness of \texttt{POMWU} in adapting to time-varying payoffs is clear, especially when compared to the classic \texttt{OMWU}, whose contextual regret grows linearly due to its inability to account for states of nature. Interestingly, \Cref{fig:proportion} shows that rounds where contexts are mispredicted contributes to a large and growing share of regret over time for \texttt{POMWU}. This illustrates the convergence of the algorithm on each context. The fact that average prediction error of the online logistic regression (\Cref{fig:error}) decreases at a slow rate thus explains most of the regret trend of \texttt{POMWU} in late rounds. Finally, \Cref{fig:equilibrium} depicts the average proportion of agents occupying each edge of the network in different contexts under the empirical policy $\hat{\bldnu}$ defined in \Cref{proposition:noregretcoarsecorrelatedequilibrium}. By  \Cref{proposition:noregretcoarsecorrelatedequilibrium}, this is a depiction of a $10^{-3}$-approximate coarse correlated equilibrium of the game. 
\section{Conclusion}The recent extension of uncoupled learning to time-varying games marks a significant progress, as it enables the modeling of non-stationary payoff environments. However, existing literature overlooks the fact that they may be able to forecast future variations of the game. In this work, we introduce prediction-aware learning, a framework in which agents can leverage predictions about future payoffs to inform their strategies. Specifically, we propose the \texttt{POMWU} algorithm, inspired by the classic \texttt{OMWU} approach, which incorporates the predicted state of nature into the optimism step. We provide explicit guarantees on both individual regrets and social welfare, and demonstrate the effectiveness of \texttt{POMWU} in a simulated contextual game.

We believe that these findings provide a strong foundation for incorporating predictive capabilities into dynamic game-theoretic settings, with significant implications for strategic decision-making in economic and industrial applications. There are several avenues for future work to improve and expand upon this framework. First, it would be valuable to weaken the feedback provided to players—for instance, by restricting it to bandit feedback—and analyze the impact on theoretical guarantees. Second, extending the model to accommodate an infinite number of contexts presents a challenging but important direction. Finally, exploring how collaborative inference influences the game dynamics and designing algorithms that account for this interplay remains an essential question from a game-theoretic perspective.

\section*{Impact Statement}
This paper presents work whose goal is to advance the understanding of multi-agent systems. There are many potential societal consequences 
of our work, none which we feel must be specifically highlighted here.

\section*{Acknowledgment}
Funded by the European Union (ERC, Ocean, 101071601). Views and opinions expressed are however those of the author(s) only and do not necessarily reflect those of the European Union or the
European Research Council Executive Agency. Neither the European Union nor the granting authority can be held responsible for them.

\bibliographystyle{plainnat}
\bibliography{sample}

\begin{thebibliography}{72}
\providecommand{\natexlab}[1]{#1}
\providecommand{\url}[1]{\texttt{#1}}
\expandafter\ifx\csname urlstyle\endcsname\relax
  \providecommand{\doi}[1]{doi: #1}\else
  \providecommand{\doi}{doi: \begingroup \urlstyle{rm}\Url}\fi

\bibitem[Anagnostides et~al.(2022{\natexlab{a}})Anagnostides, Daskalakis, Farina, Fishelson, Golowich, and Sandholm]{anagnostides2022near}
Ioannis Anagnostides, Constantinos Daskalakis, Gabriele Farina, Maxwell Fishelson, Noah Golowich, and Tuomas Sandholm.
\newblock Near-optimal no-regret learning for correlated equilibria in multi-player general-sum games.
\newblock In \emph{Proceedings of the 54th Annual ACM SIGACT Symposium on Theory of Computing}, pages 736--749, 2022{\natexlab{a}}.

\bibitem[Anagnostides et~al.(2022{\natexlab{b}})Anagnostides, Farina, Kroer, Lee, Luo, and Sandholm]{anagnostides2022uncoupled}
Ioannis Anagnostides, Gabriele Farina, Christian Kroer, Chung-Wei Lee, Haipeng Luo, and Tuomas Sandholm.
\newblock Uncoupled learning dynamics with o (log t) swap regret in multiplayer games.
\newblock \emph{Advances in Neural Information Processing Systems}, 35:\penalty0 3292--3304, 2022{\natexlab{b}}.

\bibitem[Anagnostides et~al.(2022{\natexlab{c}})Anagnostides, Farina, Panageas, and Sandholm]{anagnostides2022optimisticmirrordescentconverges}
Ioannis Anagnostides, Gabriele Farina, Ioannis Panageas, and Tuomas Sandholm.
\newblock Optimistic mirror descent either converges to nash or to strong coarse correlated equilibria in bimatrix games, 2022{\natexlab{c}}.
\newblock URL \url{https://arxiv.org/abs/2203.12074}.

\bibitem[Anagnostides et~al.(2024)Anagnostides, Panageas, Farina, and Sandholm]{anagnostides2024convergence}
Ioannis Anagnostides, Ioannis Panageas, Gabriele Farina, and Tuomas Sandholm.
\newblock On the convergence of no-regret learning dynamics in time-varying games.
\newblock \emph{Advances in Neural Information Processing Systems}, 36, 2024.

\bibitem[Auer et~al.(2002)Auer, Cesa-Bianchi, Freund, and Schapire]{auer2002nonstochastic}
Peter Auer, Nicolo Cesa-Bianchi, Yoav Freund, and Robert~E Schapire.
\newblock The nonstochastic multiarmed bandit problem.
\newblock \emph{SIAM journal on computing}, 32\penalty0 (1):\penalty0 48--77, 2002.

\bibitem[Aumann(1987)]{aumann1987correlated}
Robert~J Aumann.
\newblock Correlated equilibrium as an expression of bayesian rationality.
\newblock \emph{Econometrica: Journal of the Econometric Society}, pages 1--18, 1987.

\bibitem[Banerjee(1992)]{banerjee1992simple}
Abhijit~V Banerjee.
\newblock A simple model of herd behavior.
\newblock \emph{The quarterly journal of economics}, 107\penalty0 (3):\penalty0 797--817, 1992.

\bibitem[Bennouna et~al.(2024)Bennouna, Zhang, Amin, and Ozdaglar]{bennouna2024}
Omar Bennouna, Jiawei Zhang, Saurabh Amin, and Asuman Ozdaglar.
\newblock Addressing misspecification in contextual optimization, 2024.
\newblock URL \url{https://arxiv.org/abs/2409.10479}.

\bibitem[Besbes et~al.(2015)Besbes, Gur, and Zeevi]{besbes2015non}
Omar Besbes, Yonatan Gur, and Assaf Zeevi.
\newblock Non-stationary stochastic optimization.
\newblock \emph{Operations research}, 63\penalty0 (5):\penalty0 1227--1244, 2015.

\bibitem[Bikhchandani et~al.(1992)Bikhchandani, Hirshleifer, and Welch]{bikhchandani1992theory}
Sushil Bikhchandani, David Hirshleifer, and Ivo Welch.
\newblock A theory of fads, fashion, custom, and cultural change as informational cascades.
\newblock \emph{Journal of political Economy}, 100\penalty0 (5):\penalty0 992--1026, 1992.

\bibitem[Blum and Mansour(2007)]{blum2007external}
Avrim Blum and Yishay Mansour.
\newblock From external to internal regret.
\newblock \emph{Journal of Machine Learning Research}, 8\penalty0 (6), 2007.

\bibitem[Blum et~al.(2017)Blum, Haghtalab, Procaccia, and Qiao]{blum2017collaborative}
Avrim Blum, Nika Haghtalab, Ariel~D Procaccia, and Mingda Qiao.
\newblock Collaborative pac learning.
\newblock \emph{Advances in Neural Information Processing Systems}, 30, 2017.

\bibitem[Boursier et~al.(2022)Boursier, Perchet, and Scarsini]{boursier22a}
Etienne Boursier, Vianney Perchet, and Marco Scarsini.
\newblock Social learning in non-stationary environments.
\newblock In Sanjoy Dasgupta and Nika Haghtalab, editors, \emph{Proceedings of The 33rd International Conference on Algorithmic Learning Theory}, volume 167 of \emph{Proceedings of Machine Learning Research}, pages 128--129. PMLR, 29 Mar--01 Apr 2022.
\newblock URL \url{https://proceedings.mlr.press/v167/boursier22a.html}.

\bibitem[Chamley(2004)]{chamley2004rational}
Christophe Chamley.
\newblock \emph{Rational herds: Economic models of social learning}.
\newblock Cambridge University Press, 2004.

\bibitem[Chen and Peng(2020)]{chen2020hedging}
Xi~Chen and Binghui Peng.
\newblock Hedging in games: Faster convergence of external and swap regrets.
\newblock \emph{Advances in Neural Information Processing Systems}, 33:\penalty0 18990--18999, 2020.

\bibitem[Chiang et~al.(2012)Chiang, Yang, Lee, Mahdavi, Lu, Jin, and Zhu]{chiang2012online}
Chao-Kai Chiang, Tianbao Yang, Chia-Jung Lee, Mehrdad Mahdavi, Chi-Jen Lu, Rong Jin, and Shenghuo Zhu.
\newblock Online optimization with gradual variations.
\newblock In \emph{Conference on Learning Theory}, pages 6--1. JMLR Workshop and Conference Proceedings, 2012.

\bibitem[Christodoulou and Koutsoupias(2005)]{christodoulou2005price}
George Christodoulou and Elias Koutsoupias.
\newblock The price of anarchy of finite congestion games.
\newblock In \emph{Proceedings of the thirty-seventh annual ACM symposium on Theory of computing}, pages 67--73, 2005.

\bibitem[Dagan et~al.(2024)Dagan, Daskalakis, Fishelson, and Golowich]{dagan2024external}
Yuval Dagan, Constantinos Daskalakis, Maxwell Fishelson, and Noah Golowich.
\newblock From external to swap regret 2.0: An efficient reduction for large action spaces.
\newblock In \emph{Proceedings of the 56th Annual ACM Symposium on Theory of Computing}, pages 1216--1222, 2024.

\bibitem[Daniely and Shalev-Shwartz(2014)]{daniely2014optimal2}
Amit Daniely and Shai Shalev-Shwartz.
\newblock Optimal learners for multiclass problems.
\newblock In \emph{Conference on Learning Theory}, pages 287--316. PMLR, 2014.

\bibitem[Daniely et~al.(2014)Daniely, Sabato, Ben-David, and Shalev-Shwartz]{daniely2014}
Amit Daniely, Sivan Sabato, Shai Ben-David, and Shai Shalev-Shwartz.
\newblock Multiclass learnability and the erm principle, 2014.
\newblock URL \url{https://arxiv.org/abs/1308.2893}.

\bibitem[Daskalakis et~al.(2011)Daskalakis, Deckelbaum, and Kim]{daskalakis2011near}
Constantinos Daskalakis, Alan Deckelbaum, and Anthony Kim.
\newblock Near-optimal no-regret algorithms for zero-sum games.
\newblock In \emph{Proceedings of the twenty-second annual ACM-SIAM symposium on Discrete Algorithms}, pages 235--254. SIAM, 2011.

\bibitem[Daskalakis et~al.(2021)Daskalakis, Fishelson, and Golowich]{daskalakis2021}
Constantinos Daskalakis, Maxwell Fishelson, and Noah Golowich.
\newblock Near-optimal no-regret learning in general games.
\newblock \emph{Advances in Neural Information Processing Systems}, 34:\penalty0 27604--27616, 2021.

\bibitem[Donti et~al.(2019)Donti, Amos, and Kolter]{donti2019}
Priya~L. Donti, Brandon Amos, and J.~Zico Kolter.
\newblock Task-based end-to-end model learning in stochastic optimization, 2019.
\newblock URL \url{https://arxiv.org/abs/1703.04529}.

\bibitem[Duvocelle et~al.(2018)Duvocelle, Mertikopoulos, Staudigl, and Vermeulen]{duvocelle2018learning}
Benoit Duvocelle, Panayotis Mertikopoulos, Mathias Staudigl, and Dries Vermeulen.
\newblock Learning in time-varying games.
\newblock \emph{arXiv preprint arXiv:1809.03066}, page~17, 2018.

\bibitem[Elmachtoub and Grigas(2020)]{elmachtoub2020}
Adam~N. Elmachtoub and Paul Grigas.
\newblock Smart "predict, then optimize", 2020.
\newblock URL \url{https://arxiv.org/abs/1710.08005}.

\bibitem[Farina et~al.(2022)Farina, Kroer, Lee, and Luo]{farina2022clairvoyant}
Gabriele Farina, Christian Kroer, Chung-Wei Lee, and Haipeng Luo.
\newblock Clairvoyant regret minimization: Equivalence with nemirovski's conceptual prox method and extension to general convex games.
\newblock \emph{arXiv preprint arXiv:2208.14891}, 2022.

\bibitem[Foster and Vohra(1997)]{foster1997}
Dean~P Foster and Rakesh~V Vohra.
\newblock Calibrated learning and correlated equilibrium.
\newblock \emph{Games and Economic Behavior}, 21\penalty0 (1-2):\penalty0 40--55, 1997.

\bibitem[Foster and Vohra(1998)]{foster1998asymptotic}
Dean~P Foster and Rakesh~V Vohra.
\newblock Asymptotic calibration.
\newblock \emph{Biometrika}, 85\penalty0 (2):\penalty0 379--390, 1998.

\bibitem[Foster et~al.(2016)Foster, Li, Lykouris, Sridharan, and Tardos]{foster2016}
Dylan~J Foster, Zhiyuan Li, Thodoris Lykouris, Karthik Sridharan, and Eva Tardos.
\newblock Learning in games: Robustness of fast convergence.
\newblock \emph{Advances in Neural Information Processing Systems}, 29, 2016.

\bibitem[Freund and Schapire(1999)]{freund1999}
Yoav Freund and Robert~E Schapire.
\newblock Adaptive game playing using multiplicative weights.
\newblock \emph{Games and Economic Behavior}, 29\penalty0 (1-2):\penalty0 79--103, 1999.

\bibitem[Frongillo et~al.(2011)Frongillo, Schoenebeck, and Tamuz]{frongillo2011}
Rafael~M. Frongillo, Grant Schoenebeck, and Omer Tamuz.
\newblock Social learning in a changing world.
\newblock In Ning Chen, Edith Elkind, and Elias Koutsoupias, editors, \emph{Internet and Network Economics}, pages 146--157, Berlin, Heidelberg, 2011. Springer Berlin Heidelberg.
\newblock ISBN 978-3-642-25510-6.

\bibitem[Gogas and Papadimitriou(2021)]{papadimitriou21}
Periklis Gogas and Theophilos Papadimitriou.
\newblock {Machine Learning in Economics and Finance}.
\newblock \emph{Computational Economics}, 57\penalty0 (1):\penalty0 1--4, January 2021.
\newblock \doi{10.1007/s10614-021-10094-}.
\newblock URL \url{https://ideas.repec.org/a/kap/compec/v57y2021i1d10.1007_s10614-021-10094-w.html}.

\bibitem[Guo et~al.(2024)Guo, Xu, and Murphy]{guo2024online}
Yongyi Guo, Ziping Xu, and Susan Murphy.
\newblock Online learning in bandits with predicted context.
\newblock In \emph{International Conference on Artificial Intelligence and Statistics}, pages 2215--2223. PMLR, 2024.

\bibitem[Hall and Willett(2013)]{hall2013dynamical}
Eric Hall and Rebecca Willett.
\newblock Dynamical models and tracking regret in online convex programming.
\newblock In \emph{International Conference on Machine Learning}, pages 579--587. PMLR, 2013.

\bibitem[Hamilton(2020)]{hamilton2020time}
James~D Hamilton.
\newblock \emph{Time series analysis}.
\newblock Princeton university press, 2020.

\bibitem[Harris et~al.(2024)Harris, Wu, and Balcan]{harris2024regret}
Keegan Harris, Zhiwei~Steven Wu, and Maria-Florina Balcan.
\newblock Regret minimization in stackelberg games with side information.
\newblock \emph{arXiv preprint arXiv:2402.08576}, 2024.

\bibitem[Hart and Mas-Colell(2000{\natexlab{a}})]{hart2000}
Sergiu Hart and Andreu Mas-Colell.
\newblock A simple adaptive procedure leading to correlated equilibrium.
\newblock \emph{Econometrica}, 68\penalty0 (5):\penalty0 1127--1150, 2000{\natexlab{a}}.

\bibitem[Hart and Mas-Colell(2000{\natexlab{b}})]{hart2000simple}
Sergiu Hart and Andreu Mas-Colell.
\newblock A simple adaptive procedure leading to correlated equilibrium.
\newblock \emph{Econometrica}, 68\penalty0 (5):\penalty0 1127--1150, 2000{\natexlab{b}}.

\bibitem[Hart and Mas-Colell(2003)]{hart2003uncoupled}
Sergiu Hart and Andreu Mas-Colell.
\newblock Uncoupled dynamics do not lead to nash equilibrium.
\newblock \emph{American Economic Review}, 93\penalty0 (5):\penalty0 1830--1836, 2003.

\bibitem[Hazan and Megiddo(2007)]{hazan2007online}
Elad Hazan and Nimrod Megiddo.
\newblock Online learning with prior knowledge.
\newblock In \emph{Learning Theory: 20th Annual Conference on Learning Theory, COLT 2007, San Diego, CA, USA; June 13-15, 2007. Proceedings 20}, pages 499--513. Springer, 2007.

\bibitem[Hazan et~al.(2014)Hazan, Koren, and Levy]{hazan2014}
Elad Hazan, Tomer Koren, and Kfir~Y. Levy.
\newblock Logistic regression: Tight bounds for stochastic and online optimization, 2014.
\newblock URL \url{https://arxiv.org/abs/1405.3843}.

\bibitem[Hinton and Jordan(2024)]{jordan2024}
Geoffrey Hinton and Michael~I. Jordan.
\newblock Advancing healthcare, e-commerce, and computational analysis with ai- applications in diagnostics, market insights, and efficiency.
\newblock \emph{AlgoVista: Journal of AI and Computer Science}, 3\penalty0 (2), Nov. 2024.
\newblock URL \url{https://algovista.org/index.php/AVJCS/article/view/28}.

\bibitem[Jordan and Mitchell(2015)]{jordan15review}
Michael Jordan and T.M. Mitchell.
\newblock Machine learning: Trends, perspectives, and prospects.
\newblock \emph{Science (New York, N.Y.)}, 349:\penalty0 255--60, 07 2015.
\newblock \doi{10.1126/science.aaa8415}.

\bibitem[Kairouz et~al.(2021)Kairouz, McMahan, Avent, Bellet, Bennis, Bhagoji, Bonawitz, Charles, Cormode, Cummings, et~al.]{kairouz2021advances}
Peter Kairouz, H~Brendan McMahan, Brendan Avent, Aur{\'e}lien Bellet, Mehdi Bennis, Arjun~Nitin Bhagoji, Kallista Bonawitz, Zachary Charles, Graham Cormode, Rachel Cummings, et~al.
\newblock Advances and open problems in federated learning.
\newblock \emph{Foundations and trends{\textregistered} in machine learning}, 14\penalty0 (1--2):\penalty0 1--210, 2021.

\bibitem[Kamenica(2019)]{kamenica2019bayesian}
Emir Kamenica.
\newblock Bayesian persuasion and information design.
\newblock \emph{Annual Review of Economics}, 11\penalty0 (1):\penalty0 249--272, 2019.

\bibitem[Kamenica and Gentzkow(2011)]{kamenica2011bayesian}
Emir Kamenica and Matthew Gentzkow.
\newblock Bayesian persuasion.
\newblock \emph{American Economic Review}, 101\penalty0 (6):\penalty0 2590--2615, 2011.

\bibitem[Kirschner and Krause(2019)]{kirschner2019stochasticbanditscontextdistributions}
Johannes Kirschner and Andreas Krause.
\newblock Stochastic bandits with context distributions, 2019.
\newblock URL \url{https://arxiv.org/abs/1906.02685}.

\bibitem[Lattimore and Szepesv{\'a}ri(2020)]{lattimore2020bandit}
Tor Lattimore and Csaba Szepesv{\'a}ri.
\newblock \emph{Bandit algorithms}.
\newblock Cambridge University Press, 2020.

\bibitem[Lauffer et~al.(2023)Lauffer, Ghasemi, Hashemi, Savas, and Topcu]{lauffer2023no}
Niklas Lauffer, Mahsa Ghasemi, Abolfazl Hashemi, Yagiz Savas, and Ufuk Topcu.
\newblock No-regret learning in dynamic stackelberg games.
\newblock \emph{IEEE Transactions on Automatic Control}, 2023.

\bibitem[LeBlanc et~al.(1975)LeBlanc, Morlok, and Pierskalla]{leblanc1975efficient}
Larry~J LeBlanc, Edward~K Morlok, and William~P Pierskalla.
\newblock An efficient approach to solving the road network equilibrium traffic assignment problem.
\newblock \emph{Transportation research}, 9\penalty0 (5):\penalty0 309--318, 1975.

\bibitem[Levy et~al.(2024)Levy, P{\k{e}}ski, and Vieille]{levy2024stationary}
Rapha{\"e}l Levy, Marcin P{\k{e}}ski, and Nicolas Vieille.
\newblock Stationary social learning in a changing environment.
\newblock \emph{Econometrica}, 92\penalty0 (6):\penalty0 1939--1966, 2024.

\bibitem[Li et~al.(2010)Li, Chu, Langford, and Schapire]{li2010contextual}
Lihong Li, Wei Chu, John Langford, and Robert~E Schapire.
\newblock A contextual-bandit approach to personalized news article recommendation.
\newblock In \emph{Proceedings of the 19th international conference on World wide web}, pages 661--670, 2010.

\bibitem[Littlestone(1988)]{littlestone1988}
Nick Littlestone.
\newblock Learning quickly when irrelevant attributes abound: A new linear-threshold algorithm.
\newblock \emph{Machine learning}, 2:\penalty0 285--318, 1988.

\bibitem[Maddux and Kamgarpour(2024)]{maddux2024}
Anna~M. Maddux and Maryam Kamgarpour.
\newblock Multi-agent learning in contextual games under unknown constraints, 2024.
\newblock URL \url{https://arxiv.org/abs/2310.14685}.

\bibitem[Mossel et~al.(2020)Mossel, Mueller-Frank, Sly, and Tamuz]{mossel2020social}
Elchanan Mossel, Manuel Mueller-Frank, Allan Sly, and Omer Tamuz.
\newblock Social learning equilibria.
\newblock \emph{Econometrica}, 88\penalty0 (3):\penalty0 1235--1267, 2020.

\bibitem[Nelson et~al.(2022)Nelson, Bhattac~harjya, Gao, Liu, Bouneffouf, and Poupart]{nelson2022linearizing}
Elliot Nelson, Debarun Bhattac~harjya, Tian Gao, Miao Liu, Djallel Bouneffouf, and Pascal Poupart.
\newblock Linearizing contextual bandits with latent state dynamics.
\newblock In \emph{Uncertainty in Artificial Intelligence}, pages 1477--1487. PMLR, 2022.

\bibitem[Peng and Rubinstein(2023)]{peng2023fastswapregretminimization}
Binghui Peng and Aviad Rubinstein.
\newblock Fast swap regret minimization and applications to approximate correlated equilibria, 2023.
\newblock URL \url{https://arxiv.org/abs/2310.19647}.

\bibitem[Peng and Rubinstein(2024)]{peng2024fast}
Binghui Peng and Aviad Rubinstein.
\newblock Fast swap regret minimization and applications to approximate correlated equilibria.
\newblock In \emph{Proceedings of the 56th Annual ACM Symposium on Theory of Computing}, pages 1223--1234, 2024.

\bibitem[Piliouras et~al.(2022)Piliouras, Sim, and Skoulakis]{piliouras2022}
Georgios Piliouras, Ryann Sim, and Stratis Skoulakis.
\newblock Beyond time-average convergence: Near-optimal uncoupled online learning via clairvoyant multiplicative weights update.
\newblock \emph{Advances in Neural Information Processing Systems}, 35:\penalty0 22258--22269, 2022.

\bibitem[Rakhlin and Sridharan(2013)]{rakhlin2013online}
Alexander Rakhlin and Karthik Sridharan.
\newblock Online learning with predictable sequences.
\newblock In \emph{Conference on Learning Theory}, pages 993--1019. PMLR, 2013.

\bibitem[Rakhlin et~al.(2015)Rakhlin, Sridharan, and Tewari]{rakhlin2015online}
Alexander Rakhlin, Karthik Sridharan, and Ambuj Tewari.
\newblock Online learning via sequential complexities.
\newblock \emph{J. Mach. Learn. Res.}, 16\penalty0 (1):\penalty0 155--186, 2015.

\bibitem[Roughgarden(2015)]{roughgarden2015intrinsic}
Tim Roughgarden.
\newblock Intrinsic robustness of the price of anarchy.
\newblock \emph{Journal of the ACM (JACM)}, 62\penalty0 (5):\penalty0 1--42, 2015.

\bibitem[Roughgarden and Tardos(2002)]{roughgarden2002bad}
Tim Roughgarden and {\'E}va Tardos.
\newblock How bad is selfish routing?
\newblock \emph{Journal of the ACM (JACM)}, 49\penalty0 (2):\penalty0 236--259, 2002.

\bibitem[Sadana et~al.(2024)Sadana, Chenreddy, Delage, Forel, Frejinger, and Vidal]{sadana2024survey}
Utsav Sadana, Abhilash Chenreddy, Erick Delage, Alexandre Forel, Emma Frejinger, and Thibaut Vidal.
\newblock A survey of contextual optimization methods for decision-making under uncertainty.
\newblock \emph{European Journal of Operational Research}, 2024.

\bibitem[Sessa et~al.(2019)Sessa, Bogunovic, Kamgarpour, and Krause]{sessa2019no}
Pier~Giuseppe Sessa, Ilija Bogunovic, Maryam Kamgarpour, and Andreas Krause.
\newblock No-regret learning in unknown games with correlated payoffs.
\newblock \emph{Advances in Neural Information Processing Systems}, 32, 2019.

\bibitem[Sessa et~al.(2021)Sessa, Bogunovic, Krause, and Kamgarpour]{sessa2021}
Pier~Giuseppe Sessa, Ilija Bogunovic, Andreas Krause, and Maryam Kamgarpour.
\newblock Contextual games: Multi-agent learning with side information, 2021.
\newblock URL \url{https://arxiv.org/abs/2107.06327}.

\bibitem[Slivkins et~al.(2019)]{slivkins2019introduction}
Aleksandrs Slivkins et~al.
\newblock Introduction to multi-armed bandits.
\newblock \emph{Foundations and Trends{\textregistered} in Machine Learning}, 12\penalty0 (1-2):\penalty0 1--286, 2019.

\bibitem[Smith and S{\o}rensen(2000)]{smith2000pathological}
Lones Smith and Peter S{\o}rensen.
\newblock Pathological outcomes of observational learning.
\newblock \emph{Econometrica}, 68\penalty0 (2):\penalty0 371--398, 2000.

\bibitem[Syrgkanis et~al.(2015)Syrgkanis, Agarwal, Luo, and Schapire]{syrgkanis2015}
Vasilis Syrgkanis, Alekh Agarwal, Haipeng Luo, and Robert~E. Schapire.
\newblock Fast convergence of regularized learning in games, 2015.
\newblock URL \url{https://arxiv.org/abs/1507.00407}.

\bibitem[Yang and Ren(2021)]{yang2021bandit}
Jianyi Yang and Shaolei Ren.
\newblock Bandit learning with predicted context: Regret analysis and selective context query.
\newblock In \emph{IEEE INFOCOM 2021-IEEE Conference on Computer Communications}, pages 1--10. IEEE, 2021.

\bibitem[Zhang et~al.(2022)Zhang, Zhao, Luo, and Zhou]{zhang2022no}
Mengxiao Zhang, Peng Zhao, Haipeng Luo, and Zhi-Hua Zhou.
\newblock No-regret learning in time-varying zero-sum games.
\newblock In \emph{International Conference on Machine Learning}, pages 26772--26808. PMLR, 2022.

\bibitem[Zinkevich(2003)]{zinkevich2003online}
Martin Zinkevich.
\newblock {Online Convex Programming and Generalized Infinitesimal Gradient Ascent}.
\newblock In \emph{Machine Learning, Proceedings of the Twentieth International Conference {(ICML} 2003), August 21-24, 2003, Washington, DC, {USA}}, 2003.

\end{thebibliography}
\newpage
\appendix
\onecolumn
\section{Experiment.}\label{appendix:experiment}
\paragraph{Additional information about the setting of the experiment.}The graph used to model the Sioux Falls road network from \citet{leblanc1975efficient} has $\abs{\cN}=24$ nodes and $\abs{\cE}=76$ edges. The network topology, the cost coefficients $z_{p,\ell}>0$ as well as the quantities $q_j >0$ to be sent are downloaded from \url{https://github.com/sessap/contextualgames/tree/main/SiouxFallsNet}. In the experiment, we consider $m=5$ states of nature. Each state of nature $i\in[m]$ is generated by adding a noise $\vareps^{i}_{p,\ell}>0$ drawn from an exponential distribution with scale parameter $\lambda = 10^{-2}$ to each edge $(p,\ell)\in\cE$. For each player $j\in\agents$, we let $\cA^j$ be the $K=5$ shortest paths connecting $n_j \in\cN$ to $m_j \in\cN$. While there are $\abs{\cN}(\abs{\cN}-1)=552$ agents in total on the network, we exclude agents for whom the lengths of the longest path exceeds the length of the shortest path by more than $2$. This is because the optimal action tends to trivially be the shortest path irrespective of the state of nature for these agents. With this choice, we are left with $J=91$ agents having actions generating rewards of the same order of magnitude. The simulation is run over $T=2.10^4$ timesteps. The displayed regrets for \texttt{predMWU} and \texttt{OMWU} are averaged over agents.

\paragraph{Online supervised learning in the experiment.}As explained in the main text, for any $t\in\timesteps$,
$$
\P (Z_t=z | X_t ^0) = \zeta(\beta^\star _z , X^0)=\frac{\exp(\beta^\star _z X^0)}{\sum_{z'\in\cZ} \exp(\beta^\star _{z'}X^0)}\eqsp,
$$where $X^0 \in\R^b$ with $b=10$. In the experiment, for any $t\in\timesteps$, $X^0 _t \sim \cN (m, 5I_{b})$ with $m\in [1,4]^b$. All agents receive the same covariates from sack of simplicity. For any $z\in\cZ$, $\beta^\star _z \in\R^b$ is drawn before the simulation according to a Normal distribution $\cN (0, 5I_{b})$.

\section{Discussion of \Cref{assumption:finitecontext}.}\label{appendix:discussionfiniteness}

In this appendix, we briefly outline two possible approaches to carry our analysis without \Cref{assumption:finitecontext}. More precisely, we assume in this section that $\cZ\subset$ a compact context set.
\begin{enumerate}
    \item One first natural strategy involves discretizing $\cZ$ using an $\varepsilon$-net and projecting each incoming context $z\in\cZ$ to its nearest neighbor on the net. Under suitable smoothness conditions on the loss or reward functions, this could enable a reduction to the finite context case. However, this method typically results in regret bounds with exponential dependence on the dimension $d$, which can severely limit its applicability in high-dimensional settings (see, e.g., Theorem 4 in \citet{hazan2007online}).
    \item Alternatively, one could seek to work directly with the continuous context space. However, without placing restrictions on the policy class, it is possible to construct problem instances where both external and swap regrets grow linearly in the number of rounds. This highlights the necessity of controlling model complexity, for example by (i) restricting to a finite policy class (as in \citet{auer2002nonstochastic}), or (ii) assuming linear structure in the policies (cf. LinUCB-style approaches), possibly combined with complexity measures such as sequential Rademacher complexities \citep{
rakhlin2015online}. Each of these directions would require introducing new assumptions, proof techniques, and analytical frameworks, and therefore warrants a dedicated study.
\end{enumerate}

\section{Discussion on \Cref{definition:coarsecorrelatedeq} and \Cref{definition:correlatedeq}.}\label{appendix:equilibrium}In this section, we provide further interpretation about the equilibrium concepts defined in \Cref{definition:coarsecorrelatedeq} and \Cref{definition:correlatedeq}. First, of all, recall that $\varepsilon$-contextual coarse correlated equilibrium is defined as follows: 

\definitioncoarsecorrelatedequilibrium*

For any $z\in\cZ$, the distribution $\bldnu(z)$ can be interpreted as a correlation device that generates and recommends pure actions to agents. We say that $\bldnu(z)$ is an equilibrium in the sense of \Cref{definition:coarsecorrelatedeq} if no player can decrease their expected cost by ignoring the recommendations from $\bldnu$ \textit{before} they have even been drawn on average over time. Note that contrary to the classic coarse correlated equilibrium \citep{foster1998asymptotic}, $\bldnu: z\in\cZ \mapsto \bldnu (z)\in\mathcal{P}(\cA)$ is a \textit{policy} that maps contexts to  distributions over joint actions.

Second, a $\varepsilon$-correlated equilibrium is defined as follows. 

\definitioncorrelateqauilibrium*

Just as before, $\overline{\bldnu}$ can be regarded as a correlation device. It is an equilibrium in the sense of \Cref{definition:correlatedeq} if no player can decrease their expected cost by deviating from their recommended action \textit{after} it has been drawn, on average over time. From this point of view, being a correlated equilibrium is more demanding than a coarse correlated equilibrium. Note that \Cref{definition:correlatedeq} extends the classic correlated equilibrium notion \citep{aumann1987correlated} to the contextual case, by letting the swap functions $\varrho^j (\centraldot, z)$ depend on the state of nature.
\section{Useful algorithms.}\label{appendix:algos}
\begin{algorithm}[!ht]
\caption{Contextual Blum-Mansour algorithm.}
\label{algorithm:blummansour}
\begin{algorithmic}[1]
\STATE \textbf{Input:} a no-external regret policy $\pi^j  \in \Pi^j$.
\STATE Initialize $h^j _{k,0} = \emptyset$  for any $k\in\iint{1}{K}\eqsp.$
\FOR{each $t \in\iint{1}{T}$:}
\STATE Predict $\hZ^j _t \in\cZ$.
\STATE For every $k\in\iint{1}{K}$, get $p^j _{k,t} = \pi^j  (h^j _{k,t-1} , \hZ^j _t)$, and define $P^j _t = (p^j _{1,t}\,|\,\ldots\,|\, p^j _{K,t})\in\R^{K\times K}$.
\STATE Play $w^j _t \in\Delta_K$ such that $$P^j _t w^j_t = w^j _t \eqsp.$$
\STATE Observe $Z_t\in\cZ$ and $\Phi^j (\bw^{-j}_t)\in\R^{d\times K}\eqsp,$
\STATE Update $h ^j _{k,t} =h^j_{k,t-1}\cup\defEnsLigne{w^j _t [k]  \Phi^j (\bw^{-j}_t), Z_t}$ for ay $k\in\iint{1}{K}\eqsp.$
\ENDFOR
\end{algorithmic}
\end{algorithm}

\begin{algorithm}[!ht]
\caption{Standard Optimal Algorithm (SOA) from \citet{daniely2014}.}
\label{algorithm:daniely1}
\begin{algorithmic}[1]
\STATE \textbf{Input:} An hypothesis class $\cG \subset \defEnsLigne{g: \cX\rightarrow \cZ}$ with Littlestone dimension $\text{dim}_{\mathscr{L}}(\cG)<\infty$.
\STATE Initialize $V_0 = \cG$.
\FOR{each $t \in\iint{1}{T}$:}
\STATE Receive $X_t \in\cX$ and define $V^{(z)}_t =\defEnsLigne{g\in V_{t-1}:\: g(X_t) = z}$ for any $z\in\cZ\eqsp.$
\STATE Predict $\hZ_t \in \argmax_{z\in\cZ}\text{dim}_{\mathscr{L}}(V^{(z)}_t)\eqsp.$
\STATE Receive $Z_t \in\cZ$ and update $V_t \leftarrow V_t ^{(Z_t)}\eqsp.$
\ENDFOR
\end{algorithmic}
\end{algorithm}

\begin{algorithm}[!ht]
\caption{Learning with Expert Advice (LEA) from \citet{daniely2014}.}
\label{algorithm:daniely2}
\begin{algorithmic}[1]
\STATE \textbf{Input:} An hypothesis class $\cG \subset \defEnsLigne{g: \cX\rightarrow \cY}$ with Littlestone dimension $\text{dim}_{\mathscr{L}}(\cG)<\infty$, $N>0$ experts using \Cref{algorithm:daniely1} with $N\leq (mT)^{\text{dim}_{\mathscr{L}}(\cG)}$.
\STATE Set $\eta=\sqrt{8\ln (N) / T}\eqsp.$
\FOR{each $t \in\iint{1}{T}$:}
\STATE Observe $X_t \in\cX$, receive expert advices $(f^1 _t (X_t),\ldots, f^N _t (X_t))\in\cZ^{N}\eqsp.$
\STATE Predict $\hZ_t = f^i _t (X_t)$ with probability proportional to $\exp(-\eta \sum_{\tau < t}\2{f^i _\tau (X_\tau) \ne Z_\tau})$.
\STATE Receive $Z_t \in\cZ$ and send it to all experts as a feedback. 
\ENDFOR
\end{algorithmic}
\end{algorithm}

\begin{algorithm}[!ht]
\caption{Optimistic Mirror Descent with predicted context.}
\label{algorithm:omd}
\begin{algorithmic}[1]
\STATE Initialize $\Psi_{z_1}=\ldots=\Psi_{z_{m}}=\boldsymbol{0}_{d\times K}$ and $\rho_{z_1} = \ldots = \rho_{z_m}=\argmin_{\widetilde{w}\in\Delta_K}\cR (\widetilde{w})$.
\FOR{each $t \in\timesteps$}
\STATE Observe $\hZ^j_t \in\cZ$, set $\widetilde{M}^j _t = \Psi_{\hZ^j _t} $ and $\widetilde{g}^j _t = \rho_{\hZ^j _t}$.
\STATE Play $\widetilde{w}^j _t = \argmin_{\widetilde{w}\in\Delta_{K}}\eta\ps{\widetilde{M}^{j\top} _t \hZ^j _t}{\widetilde{w}} + D_{\cR}(\widetilde{w},\widetilde{g}^j _t)\eqsp,$
\STATE Observe $Z_t \in\R^d$ and $\Phi^{j}(\bw^{-j}_t) \in\R^{d\times K}\eqsp.$
\STATE Compute $\tilde{\rho}_t = \argmin_{g\in\Delta_{K}}\ps{\Phi^{j}(\bw^{-j}_{t})^{\top} Z_t}{g} + D_{\cR}(g, \widetilde{g}^j _t)\eqsp,$
\STATE Update $\Psi_{Z_t}\leftarrow \Phi^{j}(\bw^{-j}_t)$ and $\rho_{Z_t} \leftarrow \tilde{\rho}_t$.
\ENDFOR
\end{algorithmic}
\end{algorithm}

\begin{algorithm}[!ht]
\caption{Optimistic FTRL with predicted context.}
\label{algorithm:oftrl}
\begin{algorithmic}[1]
\STATE Initialize $w_0 = \argmin_{w\in\Delta_{K}}\cR (w)$ and  $\Psi_{z_1}=\ldots=\Psi_{z_m}=\boldsymbol{0}_{d\times K}$.
\FOR{each $t \in\timesteps$}
\STATE Observe $\hZ^j_t \in\cZ$ and set $\widetilde{M}^j _t = \Psi_{\hZ^j _t}\eqsp.$
\STATE Play $w^j _t = \argmin_{w\in\Delta_{K}}\ps{\sum_{\tau=1}^{t-1}\1\defEnsLigne{Z_{\tau}=Z^j _t}\Phi^{j}(\bw^{-j}_{\tau})^{\top} Z_\tau + \widetilde{M}^{j\top}_t Z^j _{t}}{w} + \frac{D_{\cR}(w)}{\eta}\eqsp,$
\STATE Observe $Z_t \in\R^d$ and $\Phi^{j}(\bw^{-j}_t)$, update $\Psi_{Z_t}\leftarrow\Phi^{j}(\bw^{-j}_t)\eqsp.$ 
\ENDFOR
\end{algorithmic}
\end{algorithm}
\section{Notations}\label{appendix:notations}
For the proofs, we use the following notations and shorthands.
\begin{itemize}
    \item For any $z\in\cZ$, $\scrT^z = \defEnsLigne{t\in\timesteps:\: Z_t = z}=\defEnsLigne{t^z _1 , \ldots, t^z _{n_z}}$ where $n_z = |\scrT^z |\eqsp.$
    \item For any $j\in\agents$, $z\in\cZ$ and $i\in\iint{1}{n_z}$: 
    $$
\begin{array}{ccc}
   \Phi^{j} (\bw^{-j}_{t^z _i})=\Phi^{j}_{z,i}  &\quad & w^j _{t^z _i}=w^{j}_{z,i} \\
     M^j _{t^z _i} = M^{j}_{z,i}&\quad & g^j _{t^z _i} = g^{j}_{z,i} \\
     \tilde{\rho}^{j}_{t^z _i} = \tilde{\rho}^j _{z,i} & \quad & \hZ^j _{t^z _i} = \hZ^j _{z,i}\eqsp.
\end{array}
    $$
\end{itemize}
\section{Technical lemmas.}

\lemmarewritingutility*
\begin{proof}
    Let $j\in\agents$, $Z\in\cZ$ and $\bw\in\mathscr{P}(\cA)$ with $\bw = w^j \otimes \bw^{-j}$. By Fubini theorem, 
    \begin{align*}
    c^j (\bw, Z)&= \EEs{w^j}{\EEs{\bw^{-j}}{\ps{Z}{\phi^j (a^j , \ba^{-j})}}} = \sum_{k = 1}^{K}w^j [k] \,\EEs{\bw^{-j}}{\ps{Z}{\phi_j (a^j _k, \ba^{-j})}} \\
    &= \ps{Z}{\sum_{k=1}^{K}w^j [k]\EEs{\bw^{-j}}{\phi_j (a^j _k , \ba^{-j})}} = \ps{Z}{\Phi^{j}(\bw^{-j})w^j}\eqsp.
\end{align*}
\end{proof}
\begin{lemma}\label{lemma:ohedegeequivalentomd}Let $j\in\agents$. For given sequences $(Z_1, \ldots, Z_T)\in\cZ^T$, $(\hZ^j _1, \ldots, \hZ^j _T)\in\cZ^T$ and $(\bw^{-j}_1, \ldots, \ldots, \bw^{-j}_T)\in\mathscr{P}(\cA^{-j})^T$, 
    \Cref{algorithm:ohedge} and \Cref{algorithm:omd} with $\cR : w\mapsto \sum_{k\in [K]}w[k] \ln w[k]-w[k]$ produce the same iterates: $\widetilde{w}^j_t = w^j_ t $ for any $t\in\timesteps$.
\end{lemma}
\begin{proof}Observe that the Bregman divergence $D_{\cR}(w,v)=\cR(w)-\cR(v)-\ps{\nabla\cR (v)}{w-v}$ generated by $\cR : w\mapsto \sum_{k\in [K]}w[k] \ln w[k]-w[k]$ is for any $w,v\in\Delta_K$:
$$D_{\cR}(w,v) = \KL (w,v)\eqsp,$$ where $\KL$ denotes the the Kullback-Leibler divergence. Therefore, in this proof we write $\KL$ instead of $D_{\cR}$. \Cref{algorithm:ohedge} produces an iterate $w^j _t \in\Delta_K$ such that for any $\ell\in[K]$:
\begin{equation}
    \label{eq:valuealgo1}
     w^j _{t} [\ell] = \frac{\exp\parentheseDeux{-\eta\parenthese{M^{j}_t [\ell] \hZ^j _t + \sum_{\tau=1}^{t-1}\2{Z_\tau = \hZ^j _t}\Phi^{j} (\bw^{-j}_{\tau})[\ell] \hZ^j _t  }}}{\sum_{k\in[K]}\exp\parentheseDeux{-\eta\parenthese{M^{j}_t [k] \hZ^j _t + \sum_{\tau=1}^{t-1}\2{Z_\tau = \hZ^j _t}\Phi^{j} (\bw^{-j}_{\tau})[k] \hZ^j _t }}}\eqsp.
\end{equation}We will show that the iterate of \Cref{algorithm:omd}, $\widetilde{w}^j_ {t} = \argmin_{w\in\Delta_{K}}\eta\ps{\widetilde{M}^{j\top}_{t}\hZ^j _t}{w}+\KL(w,\widetilde{g}^j _{t})$ is equal to \eqref{eq:valuealgo1}. 
To this end, we define $$\cP^j _t = \defEnsLigne{\tau\in\iint{1}{t-1} : \: Z_\tau = \hZ^j _t}\eqsp,$$and we write $\cP^j_ t = \defEnsLigne{\tau_1, \ldots, \tau_{N^j _t}}$ where $N^j _t =\sum_{\tau< t}\2{Z_\tau = \hZ^j _t}$. We prove with a recursion that for any $r\in\iint{1}{N^j _t}$, we have for any $\ell\in[K]$:
\begin{equation}
    \label{lemma1:recursionhypo}
    \tilde{g}^{j} _{\tau_r} [\ell] = \frac{\exp\parentheseDeux{-\eta \sum_{i=1}^{r-1}\Phi^{j}(\bw^{-j}_{\tau_{i}})[\ell] \hZ^j _t}}{\sum_{k\in[K]}\tilde{g}^{j} _{\tau_i} [k] \exp\parentheseDeux{-\eta \sum_{i=1}^{r-1}\Phi^{j}(\bw^{-j}_{\tau_{i}})[k]} \hZ^j _t}\eqsp,
\end{equation}
For $r=1$, by definition of \Cref{algorithm:omd},  $\tilde{g}^{j} _{\tau_1}=\argmin_{g\in\Delta_{K}}\cR(g)=m^{-1}\boldsymbol{1}_{m}$ where $\boldsymbol{1}_m  = (1,\ldots, 1)^{\top}$, so \eqref{lemma1:recursionhypo} is true by the convention $\sum_{i\in\emptyset}k_i = 0$. Suppose now that \eqref{lemma1:recursionhypo} holds true for some $r\in\iint{1}{N^{j}_{t}-1}$. By definition,
    $$
\tilde{g}^{j} _{\tau_{r+1}} = \argmin_{w\in\Delta_{K}}\eta\ps{\Phi^{j}(\bw^{-j}_{\tau_r})^{\top} \hZ^j _{\tau_{r}}}{g}+\KL(w , \tilde{g}^{j} _{\tau_{r}})\eqsp.
    $$ Equivalently, it is the solution to
$$
\min_{g\in\R^m }\max_{\lambda\in\R}\cL(g,\lambda)\quad\text{with}\quad \cL(g,\lambda)=\eta\ps{\Phi^{j}(\bw^{-j}_{\tau_r})^{\top} \hZ^j _{\tau_{r}}}{g}+\KL(g , \tilde{g}^{j} _{\tau_{r}}) + \lambda (\sum_{\ell\in[K]}g [\ell]-1)\eqsp.
$$In particular $\nabla \cL (\tilde{g}^{j}_{\tau_{r+1}},\lambda)=0$, that is for any $\ell\in[K]$:
    \begin{equation}
        \label{eq:kktomd}
        \eta \Phi^{j}(\bw^{-j}_{\tau_r})[\ell] \hZ^j _{\tau_{r}} + \ln(\tilde{g}^{j} _{\tau_{r+1}} [\ell]) - \ln (\tilde{g}^{j} _{\tau_{r}} [\ell]) + \lambda=0\quad\text{so}\quad \tilde{g}^{j} _{\tau_{r+1}} [\ell]= \tilde{g}^{j} _{\tau_{r}} [\ell] \exp\parenthese{-\eta \Phi^{j}(\bw^{-j}_{\tau_r})[\ell]\hZ^j _{\tau_{r}} - \lambda}\eqsp.
    \end{equation}Using the fact that $\sum_{k\in[K]}\tilde{g}^{j} _{\tau_{r+1}} [k] = 1$, we obtain from \eqref{eq:kktomd} $\exp(\lambda)=\sum_{k\in[K]}\tilde{g}^{j} _{\tau_{r}} [k] \exp\parentheseLigne{-\eta\Phi^{j}(\bw^{-j}_{\tau_r})[k]\hZ^j _{\tau_{r}} }$, so:
    \begin{equation}
        \tilde{g}^{j} _{\tau_{r+1}} [\ell] = \frac{\tilde{g}^{j} _{\tau_{r}} [\ell]\exp(-\eta \Phi^{j}(\bw^{-j}_{\tau_r}) [\ell] \hZ^j _{\tau_{r}})}{\sum_{k\in[K]}\tilde{g}^{j} _{\tau_{r}} [k]\exp(-\eta \Phi^{j}(\bw^{-j}_{\tau_r}) [k] \hZ^j _{\tau_{r}})}\label{lemma1:recursionok}
    \end{equation}and using the recursion assumption establishes the result. Finally, observe that
\begin{align*}
    \tilde{w}^j_ {t} &= \argmin_{w\in\Delta_{K}}\eta\ps{\widetilde{M}^{j}_{t}\hZ^j _t}{w}+\KL(w,\tilde{g}^{j} _{t})\eqsp,
    \intertext{By the same lines of computation as previously, we obtain that for any $\ell\in[K]$:}
    \tilde{w}^j _t [\ell] &= \frac{\tilde{g}^{j} _{t}\exp\parentheseDeux{-\eta \widetilde{M}^{j}_t [\ell]\hZ^j _t}}{\sum_{k\in[K]}\tilde{g}^{j} _{t}\exp\parentheseDeux{-\eta \widetilde{M}^{j}_t [k] \hZ^j _t}}\eqsp,\\
    \intertext{Finally, by definition of \Cref{algorithm:omd}, $\tilde{g}^j_t\propto \tilde{g}^j _{N^j _t}\exp(-\eta \Phi^{j}(\bw^{-j}_{\tau_{N^j_t}}))$, hence by \eqref{lemma1:recursionhypo}:}
    \tilde{w}^j _t &= \frac{\exp\parentheseDeux{-\eta \parenthese{\sum_{i=1}^{N^{j}_{t}}\Phi^{j}(\bw^{-j}_{\tau_{i}})[\ell]+\widetilde{M}^{j}_t [\ell]} \hZ^j _t}}{\sum_{k\in[K]}\tilde{g}^{j} _{\tau_i} [k] \exp\parentheseDeux{-\eta \parenthese{\sum_{i=1}^{N^{j}_{t}}\Phi^{j}(\bw^{-j}_{\tau_{i}})[k]+\widetilde{M}^{j}_t [k]} \hZ^j _t}}\eqsp,
\end{align*}Observing that $\widetilde{M}^j _t = M^j_t$ by definition for any $t\in\timesteps$, we obtain the desired result.
\end{proof}

\begin{lemma}\label{lemma:ohedgeequivalentoftlr}
For given sequences $(Z_1, \ldots, Z_T)\in\cZ^T$, $(\hZ^j _1, \ldots, \hZ^j _T)\in\cZ^T$ and $(\bw^{-j}_1, \ldots, \ldots, \bw^{-j}_T)\in\mathscr{P}(\cA^{-j})$, 
    \Cref{algorithm:ohedge} and \Cref{algorithm:oftrl} produce the same iterates.
\end{lemma}
\begin{proof}
The proof proceeds as the one of \Cref{lemma:ohedegeequivalentomd}: writing the first order condition of step 4 in \Cref{algorithm:oftrl} leads to the expression \eqref{eq:valuealgo1} of the iterate of \Cref{algorithm:ohedge}.
\end{proof}
\section{Proofs.}\label{appendix:proofs}
 \externaltoswapregret*

\begin{proof}
    Let $j\in\agents$. Assume that there exists $\pi^j \in\Pi ^j$ and $f:\N^4 \rightarrow \R_{+}$ such that the regret $\regret^j _T \in\R$ of $\pi^j$ satisfies:
    \begin{equation}
        \label{prop1:assumption}
        \regret^j _T \leq f(J,T,K,m)\eqsp.
    \end{equation}We consider the policy $\overline{\pi}^j _T \in\Pi^j$ described in \Cref{algorithm:blummansour}. In this proof, we define for any $k\in\iint{1}{K}$:
            $$ r^j _k = \sum_{t\in\timesteps}\ps{w^j _t [k]\Phi^{j}(\bw^{-j}_t)^{\top}Z_t}{p^j _{t,k}}-\min_{\pi_k : \cZ \rightarrow \Delta_K}\sum_{t\in\timesteps}\ps{w^j _t [k]\Phi^{j}(\bw^{-j}_t)^{\top}Z_t}{\pi_k (Z_t)}\eqsp,$$
    the regret incurred by $\overline{\pi} ^j$ when fed with the histories $h^j _{k, 1},\ldots, h^j _{k,T} \in\cH^j $.
The swap-regret of $\overline{\pi}^j $ reads:
    \begin{align*}
        \uregret^j _T &= \sum_{t\in\timesteps}\ps{\Phi^{j}(\bw^{-j}_t)}{w^j_ t - \lambda^j _{\star}(w^j_t , Z_t)} \\
        \intertext{Note by linearity of $c^j$ in $w^j \in \Delta_K$ (\Cref{lemma:rewritingutility}), defining $\Lambda^j _{\star}: z\in\cZ \mapsto (\lambda^j _\star ( a^j _1 , z)\,|\ldots|\,\lambda^j _{\star}(a^j_K , z))\in\{0,1\}^{K\times K}$ allows to rewrite $\lambda^j _{\star} (w^j _t , Z_t)=\Lambda^j _{\star}(Z_t)w^j _t$, hence:}
        &= \sum_{t\in\timesteps}\ps{\Phi^{j}(\bw^{-j}_t)^{\top}Z_t}{P^j _t w^j _t }-\ps{\Phi^{j}(\bw^{-j}_t)^{\top}Z_t}{ \Lambda^j _\star (Z_t)w^j _t}\qquad\text{(because $w^j _t = P^j_ t w^j _t$ )} \\
        &=\sum_{t\in\timesteps}\parentheseDeux{\sum_{k\in[K]}\ps{w^j _t [k]\, \Phi^{j}(\bw^{-j}_t)^{\top}Z_t}{p^j_{k,t}}   - \sum_{k\in[K]}\ps{w^j _t [k] \,\Phi^{j}(\bw^{-j}_t)^{\top}Z_t}{\lambda^j_\star (a^j _k , Z_t)}}\\
        & \leq \sum_{k\in[K]}r^j_{k,T} \leq Kf(J,T,K,m)\eqsp.
    \end{align*}
\end{proof}
\propboundaveragesocialwelfare* 
\begin{proof}Let $\boldsymbol{\rho}:\cZ\rightarrow\cA$ be an optimal pure strategy policy, which satisfies $\sum_{t\in\timesteps}C_t (\boldsymbol{\rho}(Z_t))=C^\star$. For any $j\in\agents$ and $z\in\cZ$, we denote by $\boldsymbol{\delta}^j_\star (Z_t)$ the distribution which puts a mass 1 on the optimal action $\boldsymbol{\rho}(Z_t)[j]$ for agent $j\in\agents$. For any $(\bw_1 , \ldots , \bw_T)\in\mathscr{P}(\cA)^T$, we have
            \begin{align*}
            \sum_{t=1}^{T}C_t (\bw_t) &= \sum_{j=1}^{J}\sum_{t=1}^{T}\ps{Z_t}{\Phi^{j} (\bw^{-j}_{t})w^j _t} \leq \sum_{j=1}^{J}\regret^j _T + \sum_{j=1}^{J}\sum_{t=1}^{T}\ps{Z_t}{\Phi^{j} (\bw^{-j}_t)\boldsymbol{\delta}_\star ^j (Z_t)} \\
            &= \sum_{j=1}^{J}\regret^j _T + \sum_{j=1}^{J}\sum_{t=1}^{T}\ps{Z_t}{\EEs{\bw^{-j} _t}{\phi_j (\rho_\star ^j (Z_t)[j] , \ba^{-j})}} \\
    &\leq \sum_{j=1}^{J}\regret^j _T + \delta T C^\star + \mu \sum_{t=1}^{T}C_t (\bw_t)\eqsp,
            \end{align*}where we used the $(\delta, \mu)$-smoothness assumption in the last line. Re-arranging the terms allows to conclude.
    \end{proof}

\epscoarsecorrelatedeq*
    \begin{proof}
Let $j\in\agents$  and $\pi^j \in\Pi^j$. By definition, for any $z\in\cZ$ such that $n_z >0$,
$$T^{-1}\sum_{t\in\timesteps}\EEs{\hat{\bldnu}_T (Z_t)}{\phi^j (\ba)}=T^{-1}\sum_{z\in\cZ}\sum_{t\in\scrT^z}n_z ^{-1}\sum_{t\in\scrT^z}\EEs{\bw_t}{\phi^j (\ba)}=T^{-1}\sum_{t\in\timesteps}\EEs{\bw_t }{\phi^j (\ba)}\eqsp.$$This observation and \Cref{lemma:rewritingutility} lead to:
        \begin{align*}
        &T^{-1}\sum_{t\in\timesteps}\parenthese{c^j (\hat{\bldnu}(Z_t), Z_t) - c^j(\pi^j (Z_t), \hat{\bldnu}^{-j}(Z_t),Z_t)}\\
        &= T^{-1}\sum_{t\in\timesteps}\EEs{\hat{\bldnu}_T (Z_t)}{\ps{Z_t}{\phi^ j (\ba_t)}} - T^{-1}\sum_{t\in\timesteps}\EEs{\pi^j (Z_t)\otimes\hat{\bldnu_T}^{-j}(Z_t)}{\ps{Z_t}{\phi^ j (a^j_t ,\ba^{-j}_t)}}\\
            &= T^{-1}\sum_{t\in\timesteps}\EEs{\bw_{t}}{\ps{Z_t}{\phi^ j (\ba_t)}}- T^{-1}\sum_{t\in\timesteps}\EEs{\pi^j (Z_t)\otimes\bw_{t}^{-j}}{\ps{Z_t}{\phi^ j (a^j_t ,\ba^{-j}_t)}}\\
            &= T^{-1}\sum_{t\in\timesteps}\ps{Z_t}{\Phi^j (\bw^{-j} _t)w^j _t}-T^{-1}\sum_{t\in\timesteps}\ps{Z_t}{\Phi^{j}(\bw^{-j}_t)\pi^j (Z_t)} \\
            &\leq T^{-1}\sum_{t\in\timesteps}\ps{Z_t}{\Phi^j (\bw^{-j} _t)w^j _t}-T^{-1}\sum_{t\in\timesteps}\ps{Z_t}{\Phi^{j}(\bw^{-j}_t)\pi^j_{\star} (Z_t)} \\
            &= T^{-1}\regret^j _T \eqsp.
        \end{align*}\end{proof}
\epscorrelatedeq*
\begin{proof}
    Let $j\in\agents$ and $\varrho^j : \cA \times \cZ \rightarrow \cA$. We have:
        \begin{align*}
            &T^{-1}\sum_{t\in\timesteps}\EEs{\hat{\nu}(Z_t)}{\ps{Z_t}{\phi^j (\ba)}-\ps{Z_t}{\phi^j (\varrho^j (a^j , Z_t),\ba^{-j})}}\\
&=T^{-1}\sum_{z\in\cZ}\sum_{t\in\scrT^z}\EEs{\bw_t}{\ps{z}{\phi^j (\ba)-\phi^j (\varrho^j (a^j , z),\ba^{-j})}}\\
            \intertext{And denoting $\widetilde{\Phi}_{z}^j (\bw^j _t) = (\E_{\bw^{-j}_t}[\phi^{j}(\varrho^{j}(a^j _{\ell},z),\ba^{-j})[r]])_{r\ell}\in\R^{d\times K}$ for any $t\in\scrT^z$, by \Cref{lemma:rewritingutility}:}
            &=T^{-1}\sum_{z\in\cZ}\sum_{t\in\scrT^z }z^{\top}\parenthese{\Phi^j (\bw^{-j} _t) - \widetilde{\Phi}^j _{z} (\bw^{-j}_t)}w^j _t \\
            \intertext{For any $z\in\cZ$,  define the matrix $B^j _z \in\defEnsLigne{0,1}^{K\times K}$ with coefficients $(B^j _z)_{k,\ell}=\2{\varrho ^j (a^j _k, z)=a^j _{\ell}}$. Observe that $\widetilde{\Phi}^j _z (\bw^{-j}_t)=\Phi^j (\bw^{-j}_t)B^j _z$, so:}
            &= T^{-1}\sum_{z\in\cZ}\sum_{t\in\scrT^z}\parenthese{z^{\top}\Phi^{j}(\bw^{-j}_t)w^j _t - z^{\top}\Phi^{j}(\bw^{-j}_t)B^j _z w^j _t}\\
            \intertext{Denoting $\tilde{w}^j _t =  B^j _z w^j _t \in\Delta_K$:}
            &= T^{-1}\sum_{z\in\cZ}\parenthese{\sum_{t\in\scrT^z}z^{\top}\Phi^{j}(\bw^{-j}_t)w^j _t -\sum_{t\in\scrT^z}z^{\top}\Phi^{j}(\bw^{-j}_t)\tilde{w}^j _t}\\
             &\leq² T^{-1}\sum_{z\in\cZ}\parenthese{\sum_{t\in\scrT^z}z^{\top}\Phi^{j}(\bw^{-j}_t)w^j _t -\sum_{t\in\scrT^z}z^{\top}\Phi^{j}(\bw^{-j}_t)\lambda^j _{\star}(w^j _t, z)}\\
             &=T^{-1}\uregret^j _T \eqsp.
        \end{align*}
\end{proof}
\propcontextualrvu*
\begin{proof}By \Cref{lemma:ohedegeequivalentomd}, it is equivalent to show the result holds when players use \Cref{algorithm:omd}  with $\cR : w \mapsto \sum_{\ell \in [K]}w_{\ell}\ln (w_{\ell}) - w_{\ell}$. We assume that this is the case for the rest of the proof. To lighten notation, we drop the tilde on $\tilde{g}$, $\tilde{M}$ and $\tilde{w}$ as compared to the pseudo-code \Cref{algorithm:omd}.

    Let $j\in\agents$. We denote by $\ell ^j _T (z) = \sum_{t\in\scrT^z}\2{\hZ^j _t \ne z}$ the number of mispredictions of player $j$ on the context $z\in\cZ$. We will prove the following inequality for any $z\in\cZ$: 
    \begin{align}
        \sum_{t\in\scrT^z }\ps{\Phi^{j} (\bw^{-j}_{t})^{\top} z}{w^j _ t - \pi^j _\star (z)} &\leq \frac{(5+\ln (K))\ell^j _T (z) + \ln (K)}{\eta} + \eta\parenthese{\sum_{i=1}^{n_z}\norminf{\parenthese{\Phi^{j}_{z,i} -\Phi^{j}_{z, i-1}}^{\top} z}^2 + 4\ell^j _T (z)} \notag \\
        &- \frac{1}{8\eta}\sum_{i=1}^{n_z}\normone{w^{j}_{z, i} - w^{j}_{z, i-1}}^2 \label{eq:regretpercontext}\eqsp.
    \end{align}
    Let $z\in\cZ$. For any $t\in\scrT^z$, the instantaneous regret decomposes as:
        \begin{equation}
            \label{eq:decompositioninstataneuousregret}
            \ps{\Phi^{j} (\bw^{-j}_t)^{\top} z}{w^j_t - \pi^j _{\star} (z)} = \underbrace{\ps{(\Phi^{j}(\bw^{-j}_t) - M^j _t)^{\top} z}{w^j_t - \tilde{\rho}_t}}_{(a)} + \underbrace{\ps{M^{j\top}_t z}{w^j _t  - \tilde{\rho}_t}}_{(b)} + \underbrace{\ps{\Phi^{j} (\bw^{-j}_t)^{\top} z}{\tilde{\rho}_t - \pi^j _{\star} (z)}}_{(c)}\eqsp,
        \end{equation}where $\tilde{\rho}_t = \argmin_{g\in\Delta_{K}}\ps{\Phi^{j}(\bw^{-j}_{t})^{\top} z}{g} + D_{\cR}(g, g^j _t)$ (see \Cref{algorithm:omd}).  We bound each of these three terms. First, because $\normone{\centraldot}$ and $\norminf{\centraldot}$ are dual,
        
\begin{equation}
    (a) \leq \norminf{(\Phi^{j} (\bw^{-j}_{t}) - M^j _t)^{\top}z} \normone{w^j _t - \tilde{\rho}_t}\eqsp.
    \label{eq:boundproduct}
\end{equation}
For the second and third term, we use the following classic lemma, whose proof relies on the definition of the Bregman divergence and the first order condition.
\begin{lemma}[\citet{rakhlin2013online}]\label{lemma:optimomd}
    let $b\in\R^m $ and $c\in\R^m$, and define $a^\star = \argmin_{a\in\R^m}\ps{a}{c} + D_{\cR}(a,b)$. Then for any $d\in\R^m$, 
    $$\ps{c}{a^\star-d}\leq D_{\cR}(d,b) - D_{\cR}(d,a^\star)-D_{\cR}(a^\star , b)$$
\end{lemma}

Since $w^j _t = \argmin_{w\in\Delta_{K}}\eta\ps{M^{j\top} _t \hZ^j _t}{w} + D_{\cR}(w,g^j _t)$, applying \Cref{lemma:optimomd} to (b) gives
                \begin{align}
                (b) &\leq  \frac{1}{\eta}\ps{M^{j\top}_t (z - \hZ^j _t )}{w^j _t - \tilde{\rho}_t} + \frac{1}{\eta} \parenthese{D_{\cR}(\tilde{\rho}_t , g^j _t) - D_{\cR}(\tilde{\rho}_t, w^j _t) - D_{\cR}(w^j _t , g^j _t)}\eqsp,\notag
                \intertext{Observe that with $\cR(p)=\sum_{\ell\in[K]}p_{\ell}\ln p_{\ell} - p_{\ell}$, we have $D_{\cR}(p,q)=\KL(p,q)$. Hence  by Pinsker's inequality,}
                &\leq \frac{1}{\eta}\ps{M^{j\top}_t (z - \hZ^j _t )}{w^j _t - \tilde{\rho}_t} +  \frac{1}{\eta}\parenthese{D_{\cR} (\tilde{\rho}_t , g^j _t)-\frac{1}{2}\parenthese{\normone{\tilde{\rho}_t - w^j _ t}^2 + \normone{w^j _ t - g^j _t}^2)}}
                \label{eq:boundonb}\eqsp. 
        \end{align}
        Likewise, $\tilde{\rho}_t = \argmin_{g\in\Delta_{K}}\ps{\Phi^{(j) \top}_{t} Z_t}{g} + D_{\cR}(g, g^j _t)$ so by \Cref{lemma:optimomd}:
        \begin{equation}
            \label{eq:boundonc}
              (c) \leq \frac{1}{\eta}\parenthese{D_{\cR}(\pi^j _{\star} (z) , g^j _t) - D_{\cR}(\pi^j _{\star} (z) , \tilde{\rho}_t) - D_{\cR}(\tilde{\rho}_t,g^j _t)}\eqsp,
        \end{equation}
Plugging \eqref{eq:boundproduct}, \eqref{eq:boundonb}, \eqref{eq:boundonc} in \eqref{eq:decompositioninstataneuousregret} and summing over $\scrT^z$ yields
        \begin{align}
    \sum_{t\in\scrT^z}\ps{\Phi^{j} (\bw^{-j}_t)^{\top} z}{w^j_t - \pi^j _{\star} (z)} 
    &\leq \sum_{t\in\scrT^z}\norminf{(\Phi^{j} (\bw^{-j}_{t})  - M^j _{t} )^{\top}z}\normone{w^j _{t} - \tilde{\rho} _{t}} - \frac{1}{2\eta}\sum_{t\in\scrT^z}\parenthese{\normone{w^j _t -\tilde{\rho}_{t}}^2+\normone{w^j _{t} - g^j _t}^2} \notag \\
    &+ \frac{1}{\eta}\sum_{t\in\scrT^z}\ps{M^{j\top}_t (z - \hZ^j _t )}{w^j _t - \tilde{\rho}_t} + \frac{1}{\eta}\sum_{t\in\scrT^z}(D_{\cR}(\pi^j _{\star} (z), g^j _t)-D_{\cR}(\pi^j _{\star} (z), \tilde{\rho}_t)) \notag \\
    \intertext{Now, since $\norminf{a}\normone{b}\leq \frac{\mu}{2}\norminf{a}^2 + \frac{1}{2\mu}\normone{b}^2$ for any $\mu > 0$,}
    &\leq \frac{\mu}{2}\sum_{t\in\scrT^z}\norminf{(\Phi^{j} (\bw^{-j}_{t})  - M^j _{t} )^{\top}z}^2 - \parenthese{\frac{1}{2\eta}-\frac{1}{2\mu}}\sum_{t\in\scrT^z}\normone{w^j _{t} - \tilde{\rho} _{t}}^2 - \frac{1}{2\eta}\sum_{t\in\scrT^z}\normone{w^j _t - g^j _t}^2 \notag \\
    &+\frac{1}{\eta}\sum_{t\in\scrT^z}\ps{M^{j\top}_t (z - \hZ^j _t )}{w^j _t - \tilde{\rho}_t} + \frac{1}{\eta}\sum_{t\in\scrT^z} D_{\cR}(\pi^j _{\star} (z), g^j _t) - D_{\cR} (\pi^j _{\star} (z) , \tilde{\rho}_t)
        \intertext{Setting $\mu=2\eta$ and noticing that $-1/2\eta < -1/4\eta$ leads to}
     &\leq \eta\underbracket[0.140ex]{\sum_{t\in\scrT^z}\norminf{(\Phi^{j} (\bw^{-j}_{t})  - M^j _{t} )^{\top}z}^2}_{\text{(i)}} - \frac{1}{4\eta}\underbracket[0.140ex]{\sum_{t\in\scrT^z}\parenthese{\normone{w^j _{t} - \tilde{\rho} _{t}}^2 + \normone{w^j _t - g^j _t}^2}}_{\text{(ii)}} \notag \\
    &+ \frac{1}{\eta}\underbracket[0.140ex]{\sum_{t\in\scrT^z} D_{\cR}(\pi^j _{\star} (z), g^j _t) - D_{\cR} (\pi^j _{\star} (z) , \tilde{\rho}_t)}_{\text{(iii)}}+\frac{1}{\eta}\underbracket[0.140ex]{\sum_{t\in\scrT^z}\ps{M^{j\top}_t (z - \hZ^j _t )}{w^j _t - \tilde{\rho}_t}}_{\text{(iv)}}\eqsp.\label{eq:threesums}
        \end{align} We now bound each sum. First for term (i), writing $\scrT^z = \defEnsLigne{t^z _1 , \ldots, t^z _{n_z}}$ (and using the shorthands defined in \Cref{appendix:notations}): \begin{align}
            \sum_{t\in\scrT ^z}\norminf{\parenthese{\Phi^{j} (\bw^{-j}_{t}) - M^j _t}^{\top}z}^2 &= \sum_{i=1}^{n_z}\norminf{\parenthese{\Phi^{j}_{z,i} -M^{j}_{z,i}}^{\top}z}^2 \eqsp,\notag\\
            \intertext{By definition of \Cref{algorithm:omd} for any $i\in\iint{1}{n_z}$ we have $M^j _{z,i}=\Phi^{j}_{z,i-1}$ if $\hZ^{j}_{z,i}  = z$, so:}
            &= \sum_{i=1}^{n_z}\2{\hZ^{j}_{z,i} = z}\norminf{\parenthese{\Phi^{j}_{z,i} -\Phi^{j}_{z, i-1}}^{\top}z}^2 + \sum_{i=1}^{n_z}\2{\hZ^{j}_{z,i} \ne z}\norminf{\parenthese{\Phi^{j}_{z,i} - M^j _{z,i}}^{\top}z}^2 \notag \\
            \intertext{By \Cref{assumption:boundedutility}:}
            &\leq \sum_{i=1}^{n_z}\2{\hZ^{j}_{z,i}=z}\norminf{\parenthese{\Phi^{j}_{z,i} -\Phi^{j}_{z, i-1}}^{\top}z}^2 + 4\sum_{i=1}^{n_z}\2{\hZ^j _{z,i} \ne z} \notag\\
&\leq\sum_{i=1}^{n_z}\norminf{\parenthese{\Phi^{j}_{z,i} -\Phi^{j}_{z, i-1}}^{\top}z}^2 + 4\ell^j_ T (z)  \label{eq:boundfirstsum}\eqsp.
        \end{align}For the term (ii), observe that for any $i\in\iint{1}{n_z}$,
        \begin{equation}\label{eq:boundwithfour}
            \normone{w^{j}_{z,i} - w^{j} _{z, i-1}}^2 \leq 4\normone{w^{j}_{z,i} - g^{j} _{z,i}}^2 + 4\normone{g^{j}_{z,i} - \tilde{\rho}^{j}_{z,i-1}}^2 + 4\normone{w^j _{z,i-1}-\tilde{\rho}^{j}_{z, i-1}}^2 \eqsp.
        \end{equation}
        We then have:
        \begin{align}
\sum_{t\in\scrT^z}\parenthese{\normone{w^j _{t} - \tilde{\rho} _{t}}^2 + \normone{w^j _t - g^j _t}^2}&=\sum_{i=1}^{n_z}\parenthese{\normone{w^{j}_{z,i} - \tilde{\rho} _{z,i}}^2 + \normone{w^{j}_{z,i} - g^{j}_{z,i}}^2}  \notag \\
&=\sum_{i=1}^{n_z}\parenthese{\normone{w^{j}_{z,i-1} - \tilde{\rho} _{z,i-1}}^2 + \normone{w^{j}_{z,i-1} - g^{j}_{z,i-1}}^2}\notag \\
&+ \underbrace{\parenthese{\normone{w^j_{z, n_z}-\tilde{\rho}_{z, n_z}}^2 - \normone{w^j _{z,0}-\tilde{\rho}_{z,0}}^2}}_{\geq 0}\notag\\
&\geq \sum_{i=1}^{n_z}\normone{w^j _{z,i}-w^j_{z, i-1}}^2-\normone{g^j_{z,i}-\tilde{\rho}^j_{z,i-1}}^2\qquad\text{(by }\eqref{eq:boundwithfour}) \notag\\
\intertext{Moreover, by definition of \Cref{algorithm:omd}, $g^{j}_{z,i} = \tilde{\rho}^{j}_{z,i-1}$ whenever $\hZ^{j}_{z,i}=z$, so:}
&\geq \frac{1}{4}\sum_{i=1}^{n_z}\normone{w^{j}_{z,i} - w^{j} _{z, i-1}}^2 - \sum_{i=1}^{n_z}\2{\hZ^{j}_{z,i}\ne z}\normone{g^{j}_{z,i} - \tilde{\rho}^{j}_{z,i-1}}^2 \notag\\
&\geq \frac{1}{4}\sum_{i=1}^{n_z}\normone{w^{j}_{z,i} - w^{j} _{z, i-1}}^2 - 4\ell^j _T (z)\label{eq:boundsecondsum}
        \end{align}Regarding the term (iii), we can use the same reasoning by writing for any $i\in\iint{1}{n_z}$:
        \begin{align}
            D_{\cR}(\pi^j _{\star} (z), g^{j}_{z,i}) - D_{\cR}(\pi^j _{\star} (z), \tilde{\rho}^{j}_{z, i}) &= D_{\cR}(\pi^j _{\star} (z), g^{j}_{z,i}) - D_{\cR}(\pi^j _{\star} (z),\tilde{\rho}^{j}_{z,i-1} ) \notag \\
            &+ D_{\cR}(\pi^j _{\star} (z), \tilde{\rho}^{j}_{z,i-1}) - D_{\cR}(\pi^j _{\star} (z), \tilde{\rho}_{z,i})\eqsp,\notag \\
            \intertext{Since $g^j _{z,i}= \tilde{\rho}^{j}_{z,i-1}$ if $\hZ^{j}_{z,i} = z$, summing over $\scrT^z $ gives:}
            \sum_{i=1}^{n_z}D_{\cR}(\pi^j _{\star} (z), g^{j}_{z,i}) - D_{\cR}(\pi^j _{\star} (z), \tilde{\rho}^{j}_{z, i}) &= \sum_{i=1}^{n_z}\2{\hZ^{j}_{z,i}\ne z}\parenthese{D_{\cR}(\pi^j _{\star} (z), g^{j}_{z,i}) - D_{\cR}(\pi^j _{\star} (z),\tilde{\rho}^{j}_{z,i-1} )} \notag\\
            &+ \sum_{i=1}^{n_z}D_{\cR}(\pi^j _{\star} (z), \tilde{\rho}^{j}_{z,i-1}) - D_{\cR}(\pi^j _{\star} (z), \tilde{\rho}_{z,i}) \notag \eqsp,
            \intertext{Observing that $0\leq D_{\cR}(p,q) \leq \ln (K)$ for any $(p,q)\in\Delta_{K} ^2$ and that the second sum is telescoping:}
            \sum_{i=1}^{n_z}D_{\cR}(\pi^j _{\star} (z), g^{j}_{z,i}) - D_{\cR}(\pi^j _{\star} (z), \tilde{\rho}^{j}_{z, i}) &\leq (\ell ^j _T (z) +1)\ln (K) \eqsp.\label{eq:boundthirdsum}
        \end{align}
        Finally for the term (iv), observe that for any $t\in\scrT^z$ we have by \Cref{assumption:boundedutility}:\begin{equation}
            \label{eq:boundfourthsum}
            \ps{M^{j\top}_t (z - \hZ^j _t )}{w^j _t - \tilde{\rho}_t}\leq 4\2{\hZ^j _t \ne z}\quad\text{so}\quad \sum_{t\in\scrT^z}\ps{M^{j\top}_t (z - \hZ^j _t )}{w^j _t - \tilde{\rho}_t} \leq 4\ell^j _t (z) \eqsp.
        \end{equation}
        
        Then, plugging \eqref{eq:boundfirstsum}, \eqref{eq:boundsecondsum}, \eqref{eq:boundthirdsum} and \eqref{eq:boundfourthsum} in \eqref{eq:threesums} establishes \cref{eq:regretpercontext}, and summing \eqref{eq:regretpercontext} over $\cZ$ gives the desired result. 
\end{proof}

\propboundsumregret*
\begin{proof}
   Our proof follows from \citet{syrgkanis2015} with our new RVU bound. Let $(t,t')\in\timesteps^2$ and $j\in\agents$. Observe that:
   \begin{align*}
   \norminf{\parenthese{\Phi^{j} (\bw^{-j}_{t}) - \Phi^{j} (\bw^{-j}_{t'})}^{\top}z} &= \max_{\ell \in [K]}\abs{\EEs{\bw^{-j}_t}{\ps{\phi_j (a_\ell , \ba^{-j}_t)}{z}}-\EEs{\bw^{-j}_{t'}}{\ps{\phi_j (a_\ell , \ba^{-j}_{t'})}{z}}}\\
   \intertext{And since $\ps{\phi_j (\ba)}{z}\leq 1$ for any $\ba\in\cA$ by \Cref{assumption:boundedutility}, with $\TV$ denoting the total variation:}
   &\leq \TV(\bw^{-j}_t , \bw^{-j}_{t'})=\TV\parenthese{\bigotimes_{k\ne j}w^{k}_t , \bigotimes_{k \ne j }w^{k}_{t'}} \\
   &\leq \sum_{k\ne j}\TV (w^k _t , w^k _{t'}) = \sum_{k\ne j}\normone{w^k _t - w^k _{t'}}\eqsp.
   \end{align*}
   Squaring the previous inequality and applying Cauchy-Schwarz leads to 
    \begin{equation}\label{boundphitoutuseul}
        \norminf{\parenthese{\Phi^{j} (\bw^{-j}_{t}) - \Phi^{j} (\bw^{-j}_{t'})}^{\top}z}^2 \leq \parenthese{\sum_{k\ne j}\normone{w^k _t - w^k _{t'}}}^2 \leq (J-1)\sum_{k\ne j}\normone{w^k _t - w^k _{t'}}^2 \eqsp, 
    \end{equation}This implies:
    \begin{equation}
\sum_{j\in\agents}\norminf{\parenthese{\Phi^{j} (\bw^{-j}_{t}) - \Phi^{j} (\bw^{-j}_{t'})}^{\top}z}^2 \leq (J-1)\sum_{j\in\agents}\sum_{k\ne j}\normone{w^k _t - w^k _{t'}}^2 = (J-1)^2 \sum_{j\in\agents}\normone{w^j _t - w^j _{t'}}^2 \label{eq:boundphiw}\eqsp.
    \end{equation}
    On the other hand, summing the RVU bounds featured in \Cref{prop:contextualrvu} over players gives:
    \begin{align*}
               \sum_{j\in\agents}\regret^j _T &\leq  \frac{(5+\ln (K))L_T + mJ\ln (K)}{\eta} + \eta\parenthese{\sum_{z\in\cZ}\sum_{i\in [n_z]}\sum_{j\in\agents}\norminf{\parenthese{\Phi^{j}_{z,i}-\Phi^{j}_{z,i-1}}^{\top}z}^2 + 4L _T} \notag\\
               &-\frac{1}{16\eta}\sum_{z\in\cZ}\sum_{i\in [n_z]}\sum_{j\in\agents}\normone{w^{j}_{z, i} - w^{j}_{z, i-1}}^2  \\
               \intertext{    Plugging \eqref{eq:boundphiw} for any $z\in\cZ$, $t=t^z _{i}$ and $t' = t^z _{i-1}$ for $i\in\iint{1}{n_z}$ gives:}
              &\leq \frac{(5+\ln (K))L_T + mJ\ln (K)}{\eta} + 4\eta L_T + \parenthese{\eta(J-1)^2 - \frac{1}{16\eta}}\sum_{z\in\cZ}\sum_{i\in [n_z]}\sum_{j\in\agents}\normone{w^{j}_{z, i} - w^{j}_{z, i-1}}^2
    \end{align*}
Then, picking $\eta = (4(J-1))^{-1}$ yields the desired result.
\end{proof}
\begin{lemma}\label{lemma:stabilityoftlr}
   If player $j\in\agents$ uses \Cref{algorithm:ohedge} with a learning rate $\eta >0$, for any $i\in\iint{1}{n_z}$:
    $$
    \normone{w^{j}_{z,i}-w^{j} _{z, i-1}}  \leq 3\eta \2{\hZ^j _t =z} +2 (1- \2{\hZ^j _t = z}) \eqsp.
    $$
\end{lemma}
\begin{proof}
    Let $j\in\agents$ and  $i\in\iint{1}{n_z}$. By \Cref{lemma:ohedgeequivalentoftlr}, it is sufficient to prove that the claim holds true for \Cref{algorithm:oftrl}. First if $\hZ^{j}_{z,i} \ne z$, $\lVert w^{j}_{z,i}-w^{j} _{z, i-1}\rVert_{1}\leq 2$. Second, assume that $\hZ^{j}_{z,i}=z$. We define for any $i' \in \iint{1}{n_z}$ $f_{i'}: w\mapsto \ps{w}{\sum_{r=1}^{i'-1}\Phi^{j \top }_{z,r} z + M^{j\top} _{z, i'}z} + \eta^{-1}\cR(w)$ and $g_{i'}:  w\mapsto \ps{w}{\sum_{r=1}^{i'}\Phi^{j \top }_{z,r} z}+\eta^{-1}\cR (w)$. Observe that for any $w\in\Delta_{K}$,
\begin{equation}\label{eq:interestingfactlemma5}
        f_{i}(w) - g_{i}(w)=\ps{w}{(M^{j}_{z,i}-\Phi^{j}_{z,i})^{\top}z}\quad\text{and}\quad f_{i }(w)-g_{i-1}(w) = \ps{w}{M^{j\top}_{z,i}z}\eqsp.
    \end{equation}
    We also define $v_{i-1} = \argmin_{v\in\Delta_{K}}g_{i-1} (v)$. We have:
    \begin{equation}
        \normone{w^{j}_{z,i} - w^{j} _{z, i-1}} \leq \normone{w^{j}_{z,i} - v_{i-1}} + \normone{v_{i-1} - w^{j} _{z, i-1}}\eqsp.\label{l5eq:decomposition}
    \end{equation}One the one hand, by  $\eta^{-1}$-strong convexity of $f_i$ with respect to $\normone{\centraldot}$, we have
    \begin{align}
        \frac{1}{2\eta}\normone{w^{j}_{z,i} - v_{i-1}} &\leq f_i (v_{i-1}) - f_i (w^{j}_{z,i}) +  \ps{\nabla f_i (w^{j}_{z,i})}{w^{j}_{z,i} - v_{i-1}}  \notag
        \intertext{And since $w^{j}_{z,i} = \argmin_{w\in\Delta_{K}}f_i (w)$ by definition in \Cref{algorithm:oftrl}, the first order condition gives:}
        \frac{1}{2\eta}\normone{w^{j}_{z,i} - v_{i-1}} &\leq f_i (v_{i-1}) - f_i (w^{j}_{z,i})\label{eq:l5ineqw}
        \intertext{Since $v_{i-1}=\argmin_{v\in\Delta_{K}}g_{i-1}(v)$, we obtain by the same reasoning,}
        \frac{1}{2\eta}\normone{ w^j _{z,i}-v_{i-1}} &\leq g_{i-1} (w^{j}_{z,i}) - g_{i-1} (v_{i-1}) \eqsp.\label{eq:l5ineqv}
        \intertext{Summing \eqref{eq:l5ineqw} with \eqref{eq:l5ineqv} and applying remark \eqref{eq:interestingfactlemma5} leads to:}
        \normone{w^j_{z,i} - v_{i-1}}^2 &\leq \eta \ps{v_{i-1} - w^{j}_{z,i}}{M^{j\top}_{z,i}z}\leq \eta\normone{w^{j}_{z,i} - v_{i-1}}\norminf{M^{j\top}_{z,i}z}\notag
        \intertext{Dividing on both sides by $\lVert w^{j}_{z,i} - v_{i-1}\rVert_{1}$ gives:}
        \normone{w^j_{z,i}-v_{i-1}} &\leq \eta\norminf{M^{j\top}_{z,i}z} \leq \eta\eqsp.\label{l5eq:boundfirstterm}
    \end{align}Similarly, it is easy to check that 
    \begin{equation*}
        \frac{1}{2\eta}\normone{v_{i-1} - w^{j} _{z, i-1}}^2 \leq f_{i-1}(v_{i-1}) - f_{i-1}(w^{j} _{z, i-1})\quad\text{and}\quad\frac{1}{2\eta}\normone{w^{j} _{z, i-1}-v_{i-1}}^2 \leq g_{i-1} (w^{j} _{z, i-1}) - g_{i-1} (v_{i-1})\eqsp.
    \end{equation*}
    So once again summing these two inequalities and making use of remark \eqref{eq:interestingfactlemma5} leads to 
    \begin{equation*}
         \normone{w^{j} _{z, i-1}-v_{i-1}}^2\leq\eta\ps{v_{i-1} - w^{j} _{z, i-1}}{(M^{j}_{z,i}-\Phi^{j}_{z,i})^{\top}z}\leq \eta\normone{w^{j} _{z, i-1} - v_{i-1}}\norminf{(M^{j}_{z,i}-\Phi^{j}_{z,i})^{\top}z}\eqsp,
    \end{equation*}Dividing both sides by $\normone{w^{j} _{z, i-1}-v_{i-1}}$ gives:
    \begin{equation}
        \normone{w^{j} _{z, i-1}-v_{i-1}} \leq \eta \norminf{(M^{j}_{z,i}-\Phi^{j}_{z,i})^{\top}z} \leq 2\eta\eqsp.\label{l5eq:boundsecondterm}
    \end{equation}Finally, plugging \eqref{l5eq:boundfirstterm} and \eqref{l5eq:boundsecondterm} in \eqref{l5eq:decomposition} yields the result. 
\end{proof}

\propconvergenceindividualutility*

\begin{proof}
   Let $j\in\agents$. By \Cref{prop:contextualrvu} we know that
   \begin{equation}
       \label{prop5:contextualrvu}
              \regret^j _T \leq  \frac{(5+\ln (K))L^j _T + m\ln (K)}{\eta} + \eta\parenthese{\sum_{z\in\cZ}\sum_{i\in [n_z]}\norminf{\parenthese{\Phi^{j}_{z,i} -\Phi^{j}_{z, i-1}}^{\top}z}^2 + 4L^j _T} \eqsp. 
   \end{equation}Moreover, we proved in \eqref{boundphitoutuseul} that for any $z\in\cZ$ and $i\in\iint{1}{n_z}$, $\lVert (\Phi^{j}_{z,i} -\Phi^{j}_{z, i-1})^{\top}z\rVert^2_{\infty}\leq (J-1) \sum_{k\ne j}\lVert w^{k}_{z,i} - w^{k} _{z, i-1}\rVert^2 _1$, so summing over contexts and timesteps gives:
   \begin{align}
       \sum_{z\in\cZ}\sum_{i\in [n_z]}\norminf{ (\Phi^{j}_{z,i} -\Phi^{j}_{z, i-1})^{\top}z}^2 &\leq (J-1)\sum_{z\in\cZ}\sum_{i\in [n_z]}\parenthese{\sum_{k\ne j}\normone{w^{k}_{z,i} - w^{k}_{z,i-1}}^2} \notag\\
       \intertext{Applying \Cref{lemma:stabilityoftlr} yields:}
       &\leq (J-1)\sum_{z\in\cZ}\sum_{i\in [n_z]}\sum_{k\ne j}\parenthese{3\eta \2{\hZ^{k}_{z,i}=z} + 2\2{\hZ^{k}_{z,i}\ne z}}^2 \notag\\
       &= (J-1)\sum_{z\in\cZ}\sum_{k\ne j}\parenthese{\sum_{i\in [n_z]}9\eta^2 \2{\hZ^{k}_{z,i}=z} + 4\2{\hZ^{k}_{z,i}\ne z}} \notag\\ 
       & \leq (J-1)\sum_{k\ne j} \parenthese{\sum_{z\in\cZ}9 n_z \eta^2 + 4 \ell^k _T (z)} \notag \notag\\
       \intertext{Therefore,}
       \sum_{z\in\cZ}\sum_{i\in [n_z]}\norminf{ (\Phi^{j}_{z,i} -\Phi^{j}_{z, i-1})^{\top}z}^2&\leq(J-1)\sum_{k\ne j}(9T\eta^2 + 4L^k _T) \leq (J-1)^2 (9T\eta^2 + 4\overline{L}_T )\eqsp.  \label{prop3:boundphi}
   \end{align}Plugging \eqref{prop3:boundphi} into \eqref{prop5:contextualrvu} establishes the first part of the proposition. For the second part of the proposition, define for any $\eta >0$:
   \begin{align*}
h(\eta)&= \frac{(5+\ln (K))\overline{L}_T + m\ln (K)}{\eta} + \eta \parentheseDeux{(J-1)^2 (9T\eta^2 + 4 \overline{L}_T) + 4 \overline{L}_T }\\
&= \frac{a}{\eta} + b\eta^3 +c\eta\quad\text{with}\quad\begin{cases}
    a = (5+\ln (K) )\overline{L}_T + m\ln (K) \\
    b = 9(J-1)^2 T \\
    c = 4[(J -1)^2  +1] \overline{L}_T  \eqsp.
\end{cases}\eqsp.       
   \end{align*}
We are looking for a minimizer of $h$ to make the bound tight. Since $h$ is continuous and $\lim_{\eta\rightarrow 0}h(\eta) = \lim_{\eta \rightarrow \infty}h(\eta)=+\infty$, it admits a minimum on $(0,\infty)$, which is also unique by strict convexity. By the first order condition, $h$ is minimized for $$ \eta^\star = \sqrt{\frac{\sqrt{12ab + c^2}-c}{6b}}\eqsp.$$We now determine the order of magnitude of $\eta^\star$. On the one hand, by sub-additivity of $x\mapsto \sqrt{x}$:
\begin{equation}\label{prop6:upperboundetastar}
    \eta^\star \leq (12ab)^{1/4}(6b)^{-1/2}=\bigO ((ab)^{1/4}b^{-1/2})\eqsp.
\end{equation}On the other hand, observe that by assumption $ T=\Omega(J^2 \overline{L}_T)$, so $ab=\Omega(c^2)$. Consequently, for $T>0$ large enough there exists $\gamma > 0$ such that $12ab \geq \gamma c^2$ and it follows that
$$
    \sqrt{c^2 + 12ab}-c = \int_{c^2}^{c^2 + 12ab}\frac{dt}{2\sqrt{t}}\geq\frac{12ab}{2\sqrt{c^2 + 12ab}}\geq \frac{12ab}{2\sqrt{1+\gamma^{-1}}\sqrt{ 12ab}}\geq \frac{\sqrt{12ab}}{2\sqrt{1+\gamma^{-1}}}\eqsp,
$$so we deduce that for $T>0$ large enough, 
$$
\eta^\star \geq \sqrt{\frac{\sqrt{12ab}}{12\sqrt{1+\gamma^{-1}}b}}\quad\text{that is}\quad \eta^\star =\Omega((ab)^{1/4}b^{-1/2})\eqsp.
$$Therefore, $\eta^\star = \Theta((ab)^{1/4}b^{-1/2})$. Plugging this value in $h$ finally gives: $$ h(\eta^\star) =\Theta(b^{1/4}a^{3/4} + a^{1/4}cb^{-1/4})=\bigO(b^{1/4}a^{3/4})\eqsp,$$because $c=\bigO (a^{1/2}b^{1/2})$ by assumption. Replacing $a$ and $b$ with their actual values yields the claimed bound. 
   \end{proof}

\robustness*

\begin{proof}
    Let $j\in\agents$ and $(\bw^{-j}_1 , \ldots, \bw^{-j}_T) \in\mathscr{P}(\cA^{-j})^T$ be any sequence of competitor strategies. We have for any $(t,t')\in\timesteps^2$ and $z\in\cZ$:
    \begin{align*}
        \norminf{\parenthese{\Phi^{j}(\bw^{-j}_t) - \Phi^{j}(\bw^{-j}_{t'})}^{\top}z}^2 &\leq 2\norminf{\Phi^{j}(\bw^{-j}_t)^{\top}z}^2 + 2\norminf{\Phi^{j}(\bw^{-j}_{t'})^{\top}z}^2 \\
        &\leq 2\parenthese{\max_{k\in[d]}\ps{\EEs{\bw^{-j}_{t}}{\phi^j (a^j _\ell , \ba^{-j})}}{z}}^{2} + 2\parenthese{\max_{k\in[d]}\ps{\EEs{\bw^{-j}_{t'}}{\phi^j (a^j _\ell , \ba^{-j})}}{z}}^{2} \\
        &= 2\parenthese{\max_{k\in[d]}\EEs{\bw^{-j}_{t}}{\ps{\phi^j (a^j _\ell , \ba^{-j})}{z}}}^{2} + 2\parenthese{\max_{k\in[d]}\EEs{\bw^{-j}_{t'}}{\ps{\phi^j (a^j _\ell , \ba^{-j})}{z}}}^{2} \\
        &\leq 4\eqsp,
    \end{align*}where we have used \Cref{assumption:boundedutility} in the last line. Therefore by \Cref{prop:contextualrvu}, we have:
    $$
\regret^j _T \leq \frac{(5+\ln (K))L^j _T + m\ln (K)}{\eta} + 4\eta(T+L^j _T)=\bigO \parenthese{\frac{\ln(K)(L^j _T + m)}{\eta} + \eta (L^j _T +T)}\eqsp.
    $$Then, setting $\eta = \Theta ([\ln(K)(L^j _T + m)]^{1/2}(L^j _T +T)^{-1/2})$ leads to 
    $$
    \regret^j _T = \bigO ([\ln(K)(L^j _T + m)]^{j 1/2}(L^j _T +T)^{1/2})\eqsp.
    $$
\end{proof}

\begin{proposition}
    \label{prop:imrovementregretbound}
     Suppose that for any $t\in\timesteps$, there exists $\hZ_t \in\cZ$ such that $\hZ^j _t = \hZ_t$ for any $j\in\agents$, and let $\underline{L}_T = \sum_{t\in\timesteps}\2{\hZ_t \ne Z_t}$. Assume \Cref{assumption:boundedutility} and   \Cref{assumption:finitecontext}. If all agents use \Cref{algorithm:ohedge} with a learning rate $\eta^\star =\Theta( J^{-1/2}T^{-1/4}[\ln (K) (\underline{L}_T +m)]^{1/4})$, then:
    $$
\regret^j _T = \bigO (\,[\ln (K) (\underline{L}_T + m) ] ^{3/4} T^{1/4}J^{1/2}\,)\eqsp.
    $$
\end{proposition}
\begin{proof}In this proof, we write for any $z\in\cZ$ and $i\in\iint{1}{n_z}, \hZ_{t^z _i}=\hZ_{z,i}$. By \Cref{prop:contextualrvu} we know that
   \begin{equation}
       \label{prop8:contextualrvubis}
              \regret^j _T \leq  \frac{(5+\ln (K))\underline{L}_T + m\ln (K)}{\eta} + \eta\parenthese{\sum_{z\in\cZ}\sum_{i=1}^{n_z}\norminf{\parenthese{\Phi^{j}_{z,i} -\Phi^{j}_{z, i-1}}^{\top}z}^2 + 4\underline{L} _T} \eqsp. 
   \end{equation}For any $z\in\cZ$, we define $\underline{\ell}_{T} (z) = \sum_{t\in\scrT^z } \2{\hZ_t \ne z}$ and $\scrC^z = \defEnsLigne{i\in\iint{1}{n_z}:\: \hZ_{z,i} = z \;\text{and}\; \hZ_{z,i-1}=z}$. We have:
   \begin{align*}
       \sum_{z\in\cZ}\sum_{i=1}^{n_z}\norminf{\parenthese{\Phi^{j}_{z,i} -\Phi^{j}_{z, i-1}}^{\top}z}^2  &= \sum_{z\in\cZ}\parenthese{\sum_{i\in\scrC^z}\norminf{\parenthese{\Phi^{j}_{z,i} -\Phi^{j}_{z, i-1}}^{\top}z}^2 + \sum_{i\notin \scrC^z}\norminf{\parenthese{\Phi^{j}_{z,i} -\Phi^{j}_{z, i-1}}^{\top}z}^2} \\
       \intertext{Note that $\scrT^z \setminus \scrC^z =\defEnsLigne{i\in\iint{1}{n_z}:\: \hZ_{z,i}\ne z\;\text{or}\;\hZ_{z,i-1}\ne z}$ so $\abs{\scrT^z \setminus\scrC^{z}}\leq 2\underline{\ell}_{T} (z)$. Together with the fact that $\lVert (\Phi^j _{z,i} - \Phi^j _{z,i-1})^{\top}z\rVert \leq 4$ for any $j\in\agents$ and $i\in \scrT^z\setminus \scrC^z$, this implies:}
       &\leq \sum_{z\in\cZ}\parenthese{\sum_{i\in\scrC^z}\norminf{\parenthese{\Phi^{j}_{z,i} -\Phi^{j}_{z, i-1}}^{\top}z}^2 + 8\underline{\ell}_{T} (z)} \\
       \intertext{We proved in \eqref{boundphitoutuseul} that for any $z\in\cZ$ and $i\in\iint{1}{n_z}$, $\lVert (\Phi^{j}_{z,i} -\Phi^{j}_{z, i-1})^{\top}z\rVert^2_{\infty}\leq (J-1) \sum_{k\ne j}\lVert w^{k}_{z,i} - w^{k} _{z, i-1}\rVert^2 _1$, so}
       &\leq  \sum_{z\in\cZ }\parenthese{(J-1) \sum_{t\in\scrC^z}\sum_{k\ne j}\normone{w^{k}_ {z,i}-w^{k}_{z,i-1}}^2 + 8\underline{\ell}_{T} (z)} \\
       \intertext{And by \Cref{lemma:stabilityoftlr}:}
       &\leq \sum_{z\in\cZ}\parenthese{9(J-1)^{2}\abs{\scrC^z}\eta^2 + 8\underline{\ell}_T (z)} \\
       &\leq 9(J-1)^2 T \eta^2 + 8\underline{L}_T \eqsp.
   \end{align*}Plugging this bound in \eqref{prop8:contextualrvubis} yields:
   $$
\regret^j _T \leq \tilde{h}(\eta)=\frac{\tilde{a}}{\eta} + \tilde{b}\eta^3 + \tilde{c}\eta\quad\text{with}\quad \begin{cases}
    \tilde{a} &= (5+\ln(K))\underline{L}_T + m\ln (K) \\
    \tilde{b} &= 9(J-1)^2 T \\
    \tilde{c} &= 12 \underline{L}_T \eqsp.
\end{cases}
   $$We observe that $\tilde{c}\propto J^{-2} c$, where $c>0$ is defined in the proof of \Cref{prop:boundindividualregret}. In particular, since $\underline{L}_T \leq T$, we have $\tilde{a}\tilde{b}=\Omega(\tilde{c}^2)$, hence we do not need it as an assumption. The rest of the proof follows exactly as in \Cref{prop:boundindividualregret}.
\end{proof}

\end{document}